\documentclass[journal]{IEEEtran}
\usepackage[cmex10]{amsmath}
\usepackage{amsthm,amssymb}
\usepackage[cp1251]{inputenc}
\usepackage{graphicx}
\usepackage{epstopdf}
\usepackage{color}
\usepackage{flushend}
\usepackage{subcaption}
\usepackage{array}
\usepackage{float}                      
\usepackage[ruled,vlined,linesnumbered]{algorithm2e}
\usepackage{multirow}
\usepackage{multicol}
\usepackage{setspace}
\usepackage{tikz}
\usetikzlibrary{arrows}
\usepackage{hyperref}
\graphicspath{{./plots/}}
\theoremstyle{plain}
\newtheorem{lemma}{Lemma}
\newtheorem{definition}{Definition}
\newtheorem{proposition}{Proposition}
\newtheorem{problem}{Problem}

\newtheorem{example}{Example}
\DeclareMathOperator{\wt}{wt}

\DeclareMathOperator{\F}{\mathbb F}

\DeclareMathOperator{\bbM}{\mathbb M}

\newcommand{\call}[1]{\mbox{\texttt{#1}}}  % call function

\renewcommand{\ell}{l}

\newcommand{\set}[1]{\left\{{#1}\right\}}

\newcommand{\mA}{\mathcal{A}}

\newcommand{\mC}{\mathcal{C}}
\newcommand{\mD}{\mathcal{D}}

\newcommand{\mF}{\mathcal{F}}

\newcommand{\mI}{\mathcal{I}}

\newcommand{\mK}{\mathcal{K}}
\newcommand{\mL}{\mathcal{L}}

\newcommand{\mS}{\mathcal{S}}

\newcommand{\mV}{\mathcal{V}}
\newcommand{\mW}{\mathcal{W}}

\newcommand{\bF}{\mathbb{F}}

\newcommand{\bI}{\mathbb{I}}

\newcommand{\bK}{\mathbb{K}}

\newcommand{\bM}{\mathbb{M}}

\newcommand{\bfK}{\mathbf{K}}

\newcommand{\bfR}{\mathbf{R}}

\newcommand{\bfT}{\mathbf{T}}

\begin{document}

\title{Fast Search Method for Large Polarization Kernels}
\author{
\IEEEauthorblockN{Grigorii Trofimiuk}
%\IEEEauthorblockA{ITMO University, Russia\\
%Email: gtrofimiuk@itmo.ru}
\thanks{G. Trofimiuk is with ITMO University, Saint-Petersburg, Russia.
 E-mail: gtrofimiuk@itmo.ru
 
This work was supported by the Russian Science Foundation,
Grant No. 22-11-00208. 
 
 This work was partially presented at the International Symposium on Information Theory'2021.

 The source code is available at 
 
 \url{https://github.com/gtrofimiuk/KernelBruteforcer}.}}

\maketitle
%\markboth{IEEE Transactions on Communications,~Vol.~XX, No.~Y, July ~20YY}
%{Trofimiuk: Fast Search Method for Large Polarization Kernels}
\begin{abstract}
A novel search method for large polarization kernels is proposed. The algorithm produces a kernel with given partial distances by employing the depth-first search combined with the computation of coset leaders weight tables and sufficient conditions of code non-equivalence. Using the proposed method, we improved all existing lower bounds on the maximum error exponent for kernels of size from 17 to 29.

We also obtained kernels which admit low complexity processing by the recently proposed recursive trellis algorithm. Numerical results demonstrate the advantage of polar codes with the obtained kernels compared with shortened polar codes and polar codes with small kernels. 
\end{abstract}
\section{Introduction}
Polar codes proposed by Ar{\i}kan \cite{Arikan2009channel} achieve the symmetric capacity of a binary-input memoryless channel  $W$. They have low complexity construction, encoding, and decoding algorithms. However, their finite-length performance turns out to be quite poor.  

One way to address this problem is to replace the Ar{\i}kan matrix by a larger matrix referred to as \textit{polarization kernel} \cite{korada2010polar} having better polarization properties. Moreover, polar codes with large kernels were shown to provide an asymptotically optimal scaling exponent \cite{fazeli2018binary}.

Several fundamental bounds on the error exponent together with various kernel constructions have recently been derived \cite{korada2010polar}, \cite{presman2015binary}, \cite{lin2015linear}, \cite{trofimiuk2021shortened}. Unfortunately, the construction of a polarization kernel of an arbitrary length with the best polarization properties remains an intractable problem. In this paper we are focusing on finding the kernels with as large error exponent as possible. To do this, we propose the search algorithm for polarization kernels with given \textit{partial distances}. This algorithm is based on depth-first search combined with the computation of coset leaders weight tables and sufficient conditions of code non-equivalence. This allows us to find polarization kernels of sizes $17 \leq l \leq 29$ with a better error exponent compared to those constructed in \cite{korada2010polar}, \cite{presman2015binary}. \nocite{trofimiuk2021search}

We also demonstrate the application of the proposed approach to the construction of polarization kernels which admit low complexity processing by the recently proposed recursive trellis algorithm \cite{trifonov2021recursive, trifonov2019trellis}. Simulation results show that polar codes with the obtained large kernels provide significant performance gain compared with shortened polar codes and polar codes with small kernels. 
 
The paper is organized as follows. The background on polarization kernels is presented in Section 
\ref{sBackground}. The basic search algorithm is introduced in Section \ref{sBasicSearch}.  Section \ref{sEnhanced} presents a novel polarization kernel search method, which employs the computation of coset leader weight tables and code equivalence invariants. Application of the proposed algorithm for finding the polarization kernels with the best known error exponent is discussed in Section \ref{sMaximization}. Some polarization kernels, which admit low complexity processing, are provided in Section \ref{sKernComplexity}. Simulation results are presented in Section \ref{sNumberic}.

\section{Background}
\label{sBackground}
\subsection{Notations}
For a positive integer $n$, we denote by $[n]$ the set $\{0,1,\dots ,n-1\}$. The vector $u_a^b$ is a subvector $(u_a,u_{a+1},\dots,u_b)$ of a vector $u$. For vectors $a$ and $b$ we denote their concatenation by $a.b$. For a matrix $M$ we denote its $i$-th row by $M_i$, and $M_{i}^j$ is a submatrix of $M$, which contains rows $M_i, M_{i+1}, \cdots, M_j$. For matrices $A$ and $B$ we denote their Kronecker product by $A \otimes B$. By $A^{\otimes m}$ we denote the $m$-fold Kronecker product of matrix $A$  with itself. 
\subsection{Polarizing transformation}
\label{ss:pt}

 A \textit{polarization kernel} $K$ is a binary invertible $l \times l$ matrix, which is not upper-triangular under any column permutation \cite{korada2010polar}.  An $(n, k)$ mixed-kernel polar code \cite{presman2016mixed}, \cite{bioglio2020multikernel} is a linear block code generated by $k$ rows of matrix 
$$G_n = \mK_1 \otimes \mK_2\otimes ... \otimes \mK_s,$$ where $\mK_i$ is an $l_i \times l_i$ polarization kernel and $n = \prod_{i=1}^s l_i$.

The encoding scheme of a polar code is given by 
$$c_0^{n-1}= u_0^{n-1}G_n,$$ where $u_i,i\in \mathcal F$ are set to some pre-defined values, e.g. zero (frozen symbols),  $|\mF| = n - k$, and the remaining $k$ values $u_i$ are set to the payload data. 

\subsection{Fundamental parameters of polar codes}
\subsubsection{Error exponent}
Let $W: \{0,1\} \to \mathcal{Y}$ be a symmetric binary-input discrete memoryless channel (B-DMC) with capacity $I(W)$. By definition,
$$
I(W) = \sum_{y \in \mathcal Y} \sum_{x \in \set{0,1}} \frac{1}{2}W(y|x) \log \frac{W(y|x)}{\frac{1}{2}W(y|0)+\frac{1}{2}W(y|1)}.
$$

The Bhattacharyya parameter of $W$ is 
$$Z(W) = \sum_{y\in\mathcal{Y}} \sqrt{W(y|0)W(y|1)}.$$

Consider the polarizing transform $K^{\otimes m}$ based on the $l \times l$ kernel $K$. The matrix $K^{\otimes m}$ together with $n$ copies of $W$ induces bit subchannels $W_{m}^{(i)}(y_0^{n-1},u_0^{i-1}|u_i)$. Let 
$$Z_m^{(i)} = Z(W_{m}^{(i)}(y_0^{n-1},u_0^{i-1}|u_i))$$ be the Bhattacharyya parameter of the $i$-th subchannel, where $i$ is uniformly distributed on the set $[l^m]$. Then, for any B-DMC $W$ with $0<I(W)<1$, we say that an $\ell\times\ell$ matrix $K$ has the error exponent (also known as the rate of polarization) $E(K)$ if \cite{korada2010polar}
\begin{itemize}
\item[(i)] For any fixed $\beta < E(K)$,
%\begin{align}\label{def:rate1} 
\[
\lim_{m \to \infty} \Pr[Z_m \leq 2^{-\ell^{m\beta}}] = I(W).
\]
%\end{align}
\item[(ii)] For any fixed $\beta > E(K)$,
%\begin{align} \label{def:rate2}
\[
\lim_{m\to\infty} \Pr[Z_m \geq 2^{-\ell^{m\beta}}] = 1.
\]
%\end{align}
\end{itemize}

Suppose we construct an $(n,k)$ polar code $\mathcal C$ with kernel $K$ by selecting the frozen set $\mF$ as a set of $n-k$ indices with the highest $Z_m^{(i)}, i \in [n]$. Let $P_e(\mC, W)$ be the block error probability of $\mathcal C$ under transmission over $W$ and decoding by the successive cancellation (SC) algorithm. It was proven \cite{korada2010polar}, that if $n/k <\ I(W)$ and $\beta < E(K)$, then for sufficiently large $n$ the probability $P_e(n)$ can be bounded as
$$
P_e(\mC, W) \leq 2^{-n^\beta}.
$$ 

The error exponent is independent of $W$. Given a matrix $G$ of size $k \times n$ and rank $k$, by $\langle G \rangle$ we denote a linear code, generated by $G$. Let $d(a,b)$ be the Hamming distance between vectors $a$ and $b$. Let 
$$d(b,\mathcal C) = \min\limits_{c\in \mathcal C} d(b,c)$$ be a minimum distance between a vector $b$ and a linear  code $\mathcal C$. 

The \textit{partial distances} (PD) $\mD_i, i  \in [l]$, of the $l\times l$ kernel $K$ are defined as follows:
\begin{align}
\label{fPDDef}
\mD_i &= d(K_i,\langle K_{i+1}^{l-1}\rangle), i \in [l-1],\\
\mD_{l-1} &= d(K_{l-1},\mathbf 0).
\end{align}

The vector $\mD$ is referred to as a \textit{partial distances profile} (PDP). It was proven in \cite{korada2010polar} that for any B-DMC $W$ and any $l\times l$ polarization kernels $K$ with a PDP $\mD$, the error exponent $E(K)$ is given by 
\begin{align}
E(K) = \frac{1}{l}\sum^{l-1}_{i=0}\log_l \mD_i.
\label{f:Rate}
\end{align}

Observe that by increasing kernel dimension $l$ to infinity, it is possible to obtain $l\times l$ kernels with error exponent arbitrarily close to 1  \cite{korada2010polar}. The explicit constructions of kernels with high error exponent are provided in \cite{presman2015binary,lin2015linear}.
\subsubsection{Scaling exponent}
Another crucial polarization property is \textit{scaling exponent}. Let us fix a B-DMC channel $W$ of capacity $I(W)$ and a desired block error probability $P_e$. Given $W$ and $P_e$, suppose we wish to communicate at rate $I(W) - \Delta$ using a family of $(n,k)$ polar codes with kernel $K$. The value of $n$ scales as $O(\Delta^{- \mu(K)})$, the constant $\mu(K)$ is referred  to as the scaling exponent \cite{hassani2014finitelength}. The scaling exponent depends on the channel, and the algorithm of its computation is only known for the case of the binary erasure channel \cite{hassani2014finitelength, fazeli2014scaling}. Further results on scaling behaviour analysis are presented in \cite{guruswami2015polar,mondelli2016unified,wang2021polar,guruswami2022arikan}.

It can be proven \cite{pfister2016nearoptimal,fazeli2021binary} that there exist kernels $K^{(l)}$ of size $l$, such that $\lim\limits_{l\rightarrow \infty}\mu(K^{(l)})=2,$ i.e. the corresponding polar codes provide an optimal scaling behaviour. Constructions of the kernels with good scaling exponent are provided in \cite{yao2019explicit}.

\section{Basic search algorithm}
\label{sBasicSearch}

\subsection{Motivations of the proposed approach}
Finite length performance of polar codes can be improved by improving the polarization parameters of the constituent kernels. In this work we are focusing on the error exponent, since  it is independent of a channel and can be explicitly computed. 
\begin{problem}
\label{fExpProbl}
Find an $l \times l$ polarization kernel with the maximal error exponent, i.e. obtain \begin{equation}
\label{fExpTask}
K^{\star} = \arg \max_{K \in \bfK_l} E(K), 
\end{equation}
where $\bfK_l$ is a set of all $l \times l$ binary polarization kernels. By $E_l^{\star}$ we denote $E(K^{\star})$.
\end{problem}
Problem \ref{fExpProbl} was solved for kernels of length from $2$ to $16$ \cite{lin2015linear}. For kernels of greater size there are only various upper bounds on the error exponent \cite{korada2010polar, presman2015binary, lin2015linear}.  Obviously, the direct solution of Problem \ref{fExpProbl} is infeasible. Hence, in this work our goal is to improve the lower bound $E_l^{\star}$ by finding $l \times l$ kernels with error exponent as large as possible.
%close to $E_l^{\star}$ as possible.  

Let us denote a length $l$ PDP\ $\mD$ as \textit{valid} if there exists an $l \times l$ kernel with PDP $\mD$. Hence, the set $\bfK_l$ can be represented as the following disjoint union of sets
\begin{equation}
\bfK_l = \bigcup_{\mD \in \mathbf D_l} \bfK_l(\mD), 
\end{equation}
where $\mathbf D_l$ is a set of all valid $l$-length PDPs, and $\bfK_l(\mD)$ is a set of $l \times l$ kernels with PDP $\mD$. Such a representation of $\bfK_l$ implies the following  
approach for the approximate solution of problem \ref{fExpProbl}. Let $\widehat E_l$ be the maximal error exponent among known $l\times l$ kernels. For instance, we can combine the results from \cite{trofimiuk2021shortened} and \cite{presman2015binary}. Then

\begin{enumerate}
\item Generate the set 
$\widetilde{\mathbf D}_l$ of possibly valid PDPs $\mD$ with 
$$\frac{1}{l}\sum^{l-1}_{i=0}\log_l \mD_i >\ \widehat E_l;$$
 
\item For each $\mD \in \widetilde{\mathbf D}_l$ try to construct at least one kernel $K$ with PDP $\mD.$
\end{enumerate}
The description of the proposed solution for the second task starts in Section \ref{ssBasic} and continues in Section \ref{sEnhanced}. We used the statement  ``possibly valid" in the sense that only necessary conditions
to check whether a PDP $\mD$ is valid are known, which is described in detail in Section \ref{ssBounds}. 
\subsection{Basic algorithm}
\label{ssBasic}
For a given $l\times l$ kernel $K$ we define a sequence of 
$(l, l-\phi, d_K^{(\phi)})$ \textit{kernel codes} 
$\mC_{K}^{(\phi)} = \langle K_{\phi}^{l-1} \rangle$, $\phi \in [l], $ and $\mC_K^{(l)}$ contains only zero codeword. 

Suppose that we want to obtain an $l \times l$ polarization kernel $K$ with given PDP $\mD$. Our idea is to successively construct such a kernel starting from row $K_{l-1}$ and ending by the row $K_0$. To do this, we define the sets $\mathbb M_\phi, \phi \in [l]$ of \textit{candidate rows}, i.e. for each obtained $K$ its rows $K_{\phi} \in \mathbb M_\phi$. 

\begin{proposition}[\cite{fazeli2014scaling}]
\label{pAddRow}
Let $K$ and $\widehat K$ be $l \times l$ polarization kernels. If 
$\widehat K = XKP$,  where $X$ is a row addition matrix, corresponding to the addition of row $i$ to row $j < i$, and $P$ is a permutation matrix, then $E(\widehat K) = E(K)$ and $\mu(\widehat K) = \mu(K)$. Furthermore, $\mC_{\widehat K}^{(i)}$ is equivalent to $\mC_{K}^{(i)}$ for all $i \in [l]$, and $\mC_{\widehat K}^{(i)} = \mC_{K}^{(i)}$ if $P$ is an identity matrix.
\end{proposition}
Proposition \ref{pAddRow} implies that any kernel $K$ with PDP $\mD$ can be transformed by addition of row $i$  to row $j < i$ into a such kernel $\widehat K$ with PDP\ $\widehat \mD$ as $\wt(\widehat K_i)  = \widehat \mD_i = \mD_i$, $i \in [l].$ Therefore, without loss of generality, the sets of candidate rows can be defined as
\begin{equation}
\label{fCandSet} 
\bbM_\phi \overset{(def)}{=} \set{v_0^{l-1}| v_0^{l-1} \in \bF_2^l, \wt(v) = \mD_\phi}.
\end{equation}

\begin{algorithm}[t]
\caption{\texttt{BasicKernelSearch}$(K,\phi,\bM,\mD)$}
\label{alg_basic}
\tcp{Input: 
\\matrix $K$ containing partially constructed kernel $K_{\phi+1}^{l-1}$;\\ current row index $\phi$;\\ set of candidate rows $\bM$;\\ The desired partial distances $\mD$;}
\If{$\phi = -1$}{
        \Return $K$;\\
}
\For{\textbf{each} $v \in \bbM_\phi$}{ \label{basicLoop} 
        $w \gets d(v, \mC_{K}^{(\phi+1)})$;\tcp{Partial distance } \label{basicPDcomp}
        \If{$w = \mD_\phi$}{
                $K_{\phi} \gets v$;\\
                $\widehat K \gets $ \texttt{BasicKernelSearch}$(K, \phi-1, \bM, \mD)$;\\
                \If{$\widehat K \ne \mathbf 0^{l \times l}$}{
                        \Return $\widehat K$ \label{basicLoopReturn};\\
                }
        }
}
\Return{$\mathbf 0^{l \times l}$};\\
\tcp{Output: either kernel $K$ with the\\ required PDs $\mD$ or zero matrix $\mathbf 0^{l \times l}$\hspace{-3mm}}
\end{algorithm}

This approach can be described in terms of the depth-first search over the candidate rows, which is illustrated in Alg. \ref{alg_basic}. The search should begin by calling \textit{BasicKernelSearch}$(\mathbf 0^{l \times l}, l-1, \bM, \mD)$, depicted in Alg. \ref{alg_basic}, where $\mathbf 0^{l \times l}$ is a zero $l\times l$ matrix. At the \textit{phase} $\phi$ the algorithm iterates through 
the vectors $v \in \mathbb M_\phi$ and checks whether $d(v, \mathcal C_{K}^{(\phi+1)}) = \mD_{\phi}$. If such a row is found, the algorithm sets $K_\phi\leftarrow v$, proceeds to the phase $\phi - 1$ and starts a new search over rows $v' \in \bbM_{\phi-1}$. If there are no such rows $v'$ as $d(v', \mathcal C_{K}^{(\phi)}) = \mD_{\phi-1}$, the algorithm returns to the phase $\phi+1$ and continues iterating through the vectors $v \in \bbM_{\phi+1}$. The algorithm works until a kernel $K$ with required PDP $\mD$ is obtained or zero matrix $\mathbf 0^{l \times l}$ is returned.

Observe that the function \textit{BasicKernelSearch} induces a search tree with nodes containing partially constructed kernels $K_{\phi}^{l-1}$, which generates kernel codes $\mC_{K}^{(\phi)}$. A single node with a kernel code $\mC_{K}^{(\phi)}$ is connected with descendant nodes with kernel codes $\mC_{K}^{(\phi-1)}$ and a parent node with a kernel code $\mC_{K}^{(\phi+1)}$. Unfortunately, the size of this search tree is incredibly large, since the number of descendants of a single node is upper bounded by $|\bbM_{\phi-1}| = \binom{l}{\mD_{\phi-1}}$. Moreover, the computation of $d(v, \mC_{K}^{(\phi)})$ has the complexity of $O(2^{\min(l-\phi, \phi)})$, which makes the depth-first search on such a tree incredibly complex.  

In the following sections we propose methods, which allow to significantly reduce the size of the search tree and complexity of partial distances computation.

\section{Enhanced Search Algorithm}
\label{sEnhanced}

\subsection{Coset leader weight table}
\label{ssCLWT}

For a set $A \subseteq \F_2^n$  and a vector $v \in \F_2^n$ by $A + v$ we denote a set $\set{a + v|a \in A}$. 
For an $(n, k)$ code $\mC$ with the parity-check matrix $H$ we define a \textit{coset leader weight table} (CLWT)
\begin{equation}
\label{eqCLWT}
\bfT_{\mC}(s) = \min_{v \in \F_2^n, \ v H^T = s}\wt(v).
\end{equation}
 In other words, $\bfT_{\mC}(s)$ equals to a minimum weight of a coset $\mC + v$, where $v H^T = s$. Naturally, a linear code $\mC$ admits different parity-check matrices; therefore, we assume that a particular parity-check matrix $H$ (used in \eqref{eqCLWT}) of $\mC$ is given together with $\bfT_\mC$. We also assume that the coset representative $v$ of minimum weight is available for a particular $s$.

A CLWT can be calculated by running the Viterbi algorithm on the unexpurgated syndrome trellis of $\mC$ \cite{wolf78efficient}. This trellis consists of states labeled by syndromes, where the state $h$ at the $i$-th layer is connected to the $(i+1)$-th layer states $h$ and $h + (H^{T})_i$ by edges with $0$ and $1$ respectively. The trellis ends with $2^{n-k}$ states and the number of states is given by 
$O(n \cdot 2^{n-k})$.   
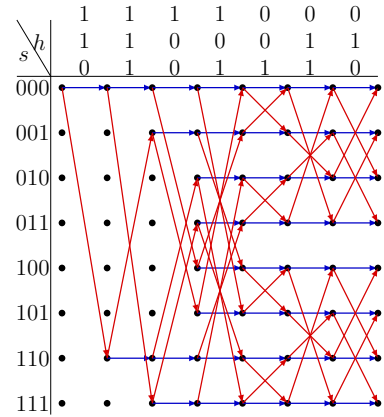
\begin{figure}[ht]
\centering
%\hspace{-260mm}
\scalebox{0.6}{
{
\Large
\usetikzlibrary{arrows}
\begin{tikzpicture}

\definecolor{darkred}{RGB}{215,0,0}
\definecolor{darkgreen}{RGB}{0,125,0}
\definecolor{neongreen}{RGB}{57,255,20}
\definecolor{darkblue}{RGB}{0,0,200}
\definecolor{deeporange}{RGB}{255,102,0}
\definecolor{darkpurple}{RGB}{136, 0, 133}
\definecolor{elpurple}{RGB}{191,0,255}
\definecolor{colblue}{RGB}{185,217,235}

\begin{scope}
\node [circle,fill,inner sep=1.5pt] (v1) at (-3.5,0) {};
\node at (-3.5,-1) [circle,fill,inner sep=1.5pt]{};
\node at (-3.5,-2) [circle,fill,inner sep=1.5pt]{};
\node at (-3.5,-3) [circle,fill,inner sep=1.5pt]{};
\node at (-3.5,-4) [circle,fill,inner sep=1.5pt]{};
\node at (-3.5,-5) [circle,fill,inner sep=1.5pt]{};
\node at (-3.5,-6) [circle,fill,inner sep=1.5pt]{};
\node at (-3.5,-7) [circle,fill,inner sep=1.5pt]{};

\node at (-3,1) {
$
\begin{matrix}
1\\
1\\
0\\
\end{matrix}
$
};
\draw[-latex, thick, color = darkblue] (-3.5,0) -- (-2.5,0);
\end{scope}

\begin{scope}[shift={(1,0)}]
\node [circle,fill,inner sep=1.5pt] (v2) at (-3.5,0) {};
\node [circle,fill,inner sep=1.5pt] at (-3.5,-1) {};
\node [circle,fill,inner sep=1.5pt] at (-3.5,-2) {};
\node [circle,fill,inner sep=1.5pt] at (-3.5,-3) {};
\node [circle,fill,inner sep=1.5pt] at (-3.5,-4) {};
\node [circle,fill,inner sep=1.5pt] at (-3.5,-5) {};
\node [circle,fill,inner sep=1.5pt] at (-3.5,-6) {};
\node [circle,fill,inner sep=1.5pt] at (-3.5,-7) {};

\node at (-3,1) {
$
\begin{matrix}
1\\
1\\
1\\
\end{matrix}
$
};
\draw[-latex, thick, color = darkblue] (-3.5,0) -- (-2.5,0);
\draw[-latex, thick, color = darkblue] (-3.5,-6) -- (-2.5,-6);
\end{scope}

\begin{scope}[shift={(2,0)}]
\node [circle,fill,inner sep=1.5pt] at (-3.5,0) {};
\node [circle,fill,inner sep=1.5pt] at (-3.5,-1) {};
\node [circle,fill,inner sep=1.5pt] at (-3.5,-2) {};
\node [circle,fill,inner sep=1.5pt] at (-3.5,-3) {};
\node [circle,fill,inner sep=1.5pt] at (-3.5,-4) {};
\node [circle,fill,inner sep=1.5pt] at (-3.5,-5) {};
\node [circle,fill,inner sep=1.5pt] at (-3.5,-6) {};
\node [circle,fill,inner sep=1.5pt] at (-3.5,-7) {};

\node at (-3,1) {
$
\begin{matrix}
1\\
0\\
0\\
\end{matrix}
$
};
\draw[-latex, thick, color = darkblue] (-3.5,0) -- (-2.5,0);
\draw[-latex, thick, color = darkblue] (-3.5,-6) -- (-2.5,-6);
\draw[-latex, thick, color = darkblue] (-3.5,-7) -- (-2.5,-7);
\draw[-latex, thick, color = darkblue] (-3.5,-1) -- (-2.5,-1);
\end{scope}

\begin{scope}[shift={(3,0)}]
\node [circle,fill,inner sep=1.5pt] at (-3.5,0) {};
\node [circle,fill,inner sep=1.5pt] at (-3.5,-1) {};
\node [circle,fill,inner sep=1.5pt] at (-3.5,-2) {};
\node [circle,fill,inner sep=1.5pt] at (-3.5,-3) {};
\node [circle,fill,inner sep=1.5pt] at (-3.5,-4) {};
\node [circle,fill,inner sep=1.5pt] at (-3.5,-5) {};
\node [circle,fill,inner sep=1.5pt] at (-3.5,-6) {};
\node [circle,fill,inner sep=1.5pt] at (-3.5,-7) {};

\node at (-3,1) {
$
\begin{matrix}
1\\
0\\
1\\
\end{matrix}
$
};
\draw[-latex, thick, color = darkblue] (-3.5,0) -- (-2.5,0);
\draw[-latex, thick, color = darkblue] (-3.5,-6) -- (-2.5,-6);
\draw[-latex, thick, color = darkblue] (-3.5,-7) -- (-2.5,-7);
\draw[-latex, thick, color = darkblue] (-3.5,-1) -- (-2.5,-1);

\draw[-latex, thick, color = darkblue] (-3.5,-2) -- (-2.5,-2);
\draw[-latex, thick, color = darkblue] (-3.5,-3) -- (-2.5,-3);
\draw[-latex, thick, color = darkblue] (-3.5,-4) -- (-2.5,-4);
\draw[-latex, thick, color = darkblue] (-3.5,-5) -- (-2.5,-5);
\end{scope}

\begin{scope}[shift={(4,0)}]
\node [circle,fill,inner sep=1.5pt] at (-3.5,0) {};
\node [circle,fill,inner sep=1.5pt] at (-3.5,-1) {};
\node [circle,fill,inner sep=1.5pt] at (-3.5,-2) {};
\node [circle,fill,inner sep=1.5pt] at (-3.5,-3) {};
\node [circle,fill,inner sep=1.5pt] at (-3.5,-4) {};
\node [circle,fill,inner sep=1.5pt] at (-3.5,-5) {};
\node [circle,fill,inner sep=1.5pt] at (-3.5,-6) {};
\node [circle,fill,inner sep=1.5pt] at (-3.5,-7) {};

\node at (-3,1) {
$
\begin{matrix}
0\\
0\\
1\\
\end{matrix}
$
};
\draw[-latex, thick, color = darkblue] (-3.5,0) -- (-2.5,0);
\draw[-latex, thick, color = darkblue] (-3.5,-6) -- (-2.5,-6);
\draw[-latex, thick, color = darkblue] (-3.5,-7) -- (-2.5,-7);
\draw[-latex, thick, color = darkblue] (-3.5,-1) -- (-2.5,-1);

\draw[-latex, thick, color = darkblue] (-3.5,-2) -- (-2.5,-2);
\draw[-latex, thick, color = darkblue] (-3.5,-3) -- (-2.5,-3);
\draw[-latex, thick, color = darkblue] (-3.5,-4) -- (-2.5,-4);
\draw[-latex, thick, color = darkblue] (-3.5,-5) -- (-2.5,-5);
\end{scope}

\begin{scope}[shift={(5,0)}]
\node [circle,fill,inner sep=1.5pt] at (-3.5,0) {};
\node [circle,fill,inner sep=1.5pt] at (-3.5,-1) {};
\node [circle,fill,inner sep=1.5pt] at (-3.5,-2) {};
\node [circle,fill,inner sep=1.5pt] at (-3.5,-3) {};
\node [circle,fill,inner sep=1.5pt] at (-3.5,-4) {};
\node [circle,fill,inner sep=1.5pt] at (-3.5,-5) {};
\node [circle,fill,inner sep=1.5pt] at (-3.5,-6) {};
\node [circle,fill,inner sep=1.5pt] at (-3.5,-7) {};

\node at (-3,1) {
$
\begin{matrix}
0\\
1\\
1\\
\end{matrix}
$
};
\draw[-latex, thick, color = darkblue] (-3.5,0) -- (-2.5,0);
\draw[-latex, thick, color = darkblue] (-3.5,-6) -- (-2.5,-6);
\draw[-latex, thick, color = darkblue] (-3.5,-7) -- (-2.5,-7);
\draw[-latex, thick, color = darkblue] (-3.5,-1) -- (-2.5,-1);

\draw[-latex, thick, color = darkblue] (-3.5,-2) -- (-2.5,-2);
\draw[-latex, thick, color = darkblue] (-3.5,-3) -- (-2.5,-3);
\draw[-latex, thick, color = darkblue] (-3.5,-4) -- (-2.5,-4);
\draw[-latex, thick, color = darkblue] (-3.5,-5) -- (-2.5,-5);
\end{scope}

\begin{scope}[shift={(6,0)}]
\node [circle,fill,inner sep=1.5pt] at (-3.5,0) {};
\node [circle,fill,inner sep=1.5pt] at (-3.5,-1) {};
\node [circle,fill,inner sep=1.5pt] at (-3.5,-2) {};
\node [circle,fill,inner sep=1.5pt] at (-3.5,-3) {};
\node [circle,fill,inner sep=1.5pt] at (-3.5,-4) {};
\node [circle,fill,inner sep=1.5pt] at (-3.5,-5) {};
\node [circle,fill,inner sep=1.5pt] at (-3.5,-6) {};
\node [circle,fill,inner sep=1.5pt] at (-3.5,-7) {};

\node at (-3,1) {
$
\begin{matrix}
0\\
1\\
0\\
\end{matrix}
$
};
\draw[-latex, thick, color = darkblue] (-3.5,0) -- (-2.5,0);
\draw[-latex, thick, color = darkblue] (-3.5,-6) -- (-2.5,-6);
\draw[-latex, thick, color = darkblue] (-3.5,-7) -- (-2.5,-7);
\draw[-latex, thick, color = darkblue] (-3.5,-1) -- (-2.5,-1);

\draw[-latex, thick, color = darkblue] (-3.5,-2) -- (-2.5,-2);
\draw[-latex, thick, color = darkblue] (-3.5,-3) -- (-2.5,-3);
\draw[-latex, thick, color = darkblue] (-3.5,-4) -- (-2.5,-4);
\draw[-latex, thick, color = darkblue] (-3.5,-5) -- (-2.5,-5);
\end{scope}

\begin{scope}[shift={(7,0)}]
\node [circle,fill,inner sep=1.5pt] at (-3.5,0) {};
\node [circle,fill,inner sep=1.5pt] at (-3.5,-1) {};
\node [circle,fill,inner sep=1.5pt] at (-3.5,-2) {};
\node [circle,fill,inner sep=1.5pt] at (-3.5,-3) {};
\node [circle,fill,inner sep=1.5pt] at (-3.5,-4) {};
\node [circle,fill,inner sep=1.5pt] at (-3.5,-5) {};
\node [circle,fill,inner sep=1.5pt] at (-3.5,-6) {};
\node [circle,fill,inner sep=1.5pt] at (-3.5,-7) {};

\end{scope}

%\node at (-7,-2.5) {
%$H = 
%\begin{pmatrix}
%1&1&1&1&0&0&0\\
%1&1&0&0&0&1&1\\
%0&1&0&1&1&1&0\\
%\end{pmatrix}
%$
%};
\node at (-4.15,0) {$000$};
\node at (-4.15,-1) {$001$};
\node at (-4.15,-2) {$010$};
\node at (-4.15,-3) {$011$};
\node at (-4.15,-4) {$100$};
\node at (-4.15,-5) {$101$};
\node at (-4.15,-6) {$110$};
\node at (-4.15,-7) {$111$};

\draw (-3.75,1.5) -- (-3.75,-7.25);
\draw (-4.5,0.25) -- (3.5,0.25);
\draw (-3.75,0.25) -- (-4.5,1.5);
\node at (-4.35,0.7) {$s$};
\node at (-4,1) {$h$};

\draw[-latex, thick, color = darkred] (-3.5,0) -- (-2.5,-6);

\draw[-latex, thick, color = darkred] (-2.5,0) -- (-1.5,-7);
\draw[-latex, thick, color = darkred] (-2.5,-6) -- (-1.5,-1);

\draw[-latex, thick, color = darkred] (-1.5,0) -- (-0.5,-4);
\draw[-latex, thick, color = darkred] (-1.5,-1) -- (-0.5,-5);
\draw[-latex, thick, color = darkred] (-1.5,-7) -- (-0.5,-3);
\draw[-latex, thick, color = darkred] (-1.5,-6) -- (-0.5,-2);

\draw[-latex, thick, color = darkred] (-0.5,0) -- (0.5,-5);
\draw[-latex, thick, color = darkred] (-0.5,-1) -- (0.5,-4);
\draw[-latex, thick, color = darkred] (-0.5,-2) -- (0.5,-7);
\draw[-latex, thick, color = darkred] (-0.5,-3) -- (0.5,-6);
\draw[-latex, thick, color = darkred] (-0.5,-4) -- (0.5,-1);
\draw[-latex, thick, color = darkred] (-0.5,-5) -- (0.5,-0);
\draw[-latex, thick, color = darkred] (-0.5,-6) -- (0.5,-3);
\draw[-latex, thick, color = darkred] (-0.5,-7) -- (0.5,-2);

\draw[-latex, thick, color = darkred] (0.5,0) -- (1.5,-1);
\draw[-latex, thick, color = darkred] (0.5,-1) -- (1.5,-0);
\draw[-latex, thick, color = darkred] (0.5,-2) -- (1.5,-3);
\draw[-latex, thick, color = darkred] (0.5,-3) -- (1.5,-2);
\draw[-latex, thick, color = darkred] (0.5,-4) -- (1.5,-5);
\draw[-latex, thick, color = darkred] (0.5,-5) -- (1.5,-4);
\draw[-latex, thick, color = darkred] (0.5,-6) -- (1.5,-7);
\draw[-latex, thick, color = darkred] (0.5,-7) -- (1.5,-6);

\draw[-latex, thick, color = darkred] (1.5,0) -- (2.5,-3);
\draw[-latex, thick, color = darkred] (1.5,-1) -- (2.5,-2);
\draw[-latex, thick, color = darkred] (1.5,-2) -- (2.5,-1);
\draw[-latex, thick, color = darkred] (1.5,-3) -- (2.5,-0);
\draw[-latex, thick, color = darkred] (1.5,-4) -- (2.5,-7);
\draw[-latex, thick, color = darkred] (1.5,-5) -- (2.5,-6);
\draw[-latex, thick, color = darkred] (1.5,-6) -- (2.5,-5);
\draw[-latex, thick, color = darkred] (1.5,-7) -- (2.5,-4);

\draw[-latex, thick, color = darkred] (2.5,0) -- (3.5,-2);
\draw[-latex, thick, color = darkred] (2.5,-1) -- (3.5,-3);
\draw[-latex, thick, color = darkred] (2.5,-2) -- (3.5,-0);
\draw[-latex, thick, color = darkred] (2.5,-3) -- (3.5,-1);
\draw[-latex, thick, color = darkred] (2.5,-4) -- (3.5,-6);
\draw[-latex, thick, color = darkred] (2.5,-5) -- (3.5,-7);
\draw[-latex, thick, color = darkred] (2.5,-6) -- (3.5,-4);
\draw[-latex, thick, color = darkred] (2.5,-7) -- (3.5,-5);
\end{tikzpicture}
}
}
\caption{Example of unexpurgated syndrome trellis}
\label{fUnTrellis}
\end{figure}
An example of an unexpurgated syndrome trellis for the linear code with the parity check matrix 
$$
H = 
\arraycolsep=2pt\def\arraystretch{0.8}
\begin{pmatrix}
1&1&1&1&0&0&0\\
1&1&0&0&0&1&1\\
0&1&0&1&1&1&0
\end{pmatrix}
$$
is provided in Fig. \ref{fUnTrellis}, where the blue edges are labeled with $0$, and the red edges are labeled with $1$.

Note that for an $(n, k+t)$ supercode $\mC' \supset \mC$ with a parity-check matrix $H'$, the CLWT $\bfT_{\mC'}$ can be recomputed from $\bfT_{\mC}$ with $(2^{n-k} - 2^{n-k-t})$ comparison operations. Namely, for a syndrome $s \in \F_2^{n-k-t}$ there are such vectors $v^{(j)}$, $j \in [2^{t}]$, as 
$$v^{(0)}{H'}^{\top} = v^{(1)}{H'}^{\top} =
 \cdots = v^{(2^{t}-1)}{H'}^{\top} = s$$
and $v^{(i)}H^{\top} \ne v^{(j)}H^{\top}$, if $i \ne j$. That is,
\begin{equation}
\label{fCLWTRecompute}
\bfT_{\mC'}(s) = \min \limits_{j\in[2^t]} \bfT_{\mC}(s^{(j)}),
\end{equation}
where $s^{(j)} = v^{(j)}H^{\top}$. 

As a result, the computation of CLWD $\bfT_\mC$ eliminates the calculation of $d(v, \mC_{K}^{(\phi+1)})$ for each $v$ in line \ref{basicPDcomp} of Alg. \ref{alg_basic}, since 
$$d(v, \mC_{K}^{(\phi+1)}) = \bfT_{\mC_{K}^{(\phi+1)}}(vH^{\top}),$$ where $H$ is a parity-check matrix of $\mC_{K}^{(\phi+1)}$.

\subsection{Chain of nested codes}
\label{ssChain}
In this section we investigate further applications of CLWT for the search of polarization kernels.
\begin{definition}
\label{dNestedCodes}
Given an $(n, k)$ linear code $\mC$ with a generator matrix $G^{(0)}$ we define a chain $\mC(V)$ of nested codes 
$$\mC_0 \subset \mC_1 \subset \cdots \subset \mC_{b+1},b < n-k,$$
such that 
\begin{equation}
\mC_{i+1} = \mC_i \cup (\mC_i+V_i),i\in[b+1],
 V_i \in (\F_2^n \setminus \mC_i),
\label{eqNestedCodes}
\end{equation}
where $\mC_0 = \mC$. A matrix $V$ of size $(b+1) \times n$ is referred to as a \textbf{coset matrix}. A generator matrix $G^{(i+1)}$ of a code $\mC_{i+1}$ is given by
$$
G^{(i+1)} = \left( \begin{array}{c}
  V_i  \\
  \hline
  G^{(i)} 
\end{array}
\right).
$$
We call the value $D_i = d(V_i, \mC_i)$ the $i$-th \textbf{coset distance}.
\end{definition}

For a matrix $V$ we define a set 
$$\mA(V_j^{i}) = \langle V_j^{i-1}\rangle + V_{i}.$$
From the construction \eqref{eqNestedCodes} of nested codes it follows that
\begin{equation}
\label{eqCosetDistance}
d(V_i, \mC_i) = \min_{v \in \mA(V_j^{i})} d(v, \mC_j),
\end{equation}
where $0 \leq j \leq i \leq b$. Note that $|\mA(V_j^{i})| = 2^{i-j}$.
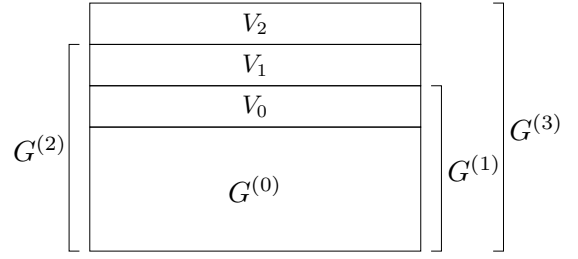
\begin{figure}[ht]
\centering
%\hspace{-260mm}
\scalebox{0.55}{
\begin{tikzpicture}

\draw  (-4.5,1.5) rectangle (3.5,-1.5);
\draw  (-4.5,2.5) rectangle (3.5,1.5);
\draw  (-4.5,3.5) rectangle (3.5,2.5);
\draw  (-4.5,4.5) rectangle (3.5,3.5);

\node[scale = 2] at (-0.5,0) {$G^{(0)}$};
\node[scale = 1.75] at (-0.5,2) {$V_0$};
\node[scale = 1.75] at (-0.5,3) {$V_1$};
\node[scale = 1.75] at (-0.5,4) {$V_2$};
\draw (3.75,2.5) -- (4,2.5) -- (4,-1.5) -- (3.75,-1.5);
\node[scale = 2]  at (4.8,0.5) {$G^{(1)}$};
\node[scale = 2]  at (6.3,1.5) {$G^{(3)}$};
\draw (-4.75,3.5) -- (-5,3.5) -- (-5,-1.5) -- (-4.75,-1.5);
\draw (5.25,4.5) -- (5.5,4.5) -- (5.5,-1.5) -- (5.25,-1.5);
\node[scale = 2] at (-5.7,1) {$G^{(2)}$};
\end{tikzpicture}
}
\caption{An example of a generator matrix of nested codes}
\label{fNestedExample}
\end{figure}
\begin{example}
Consider a chain $\mC(V)$ of nested codes $\mC = \mC_0 \subset \mC_1 
\subset \mC_2 \subset \mC_{3}$ with a nested  generator matrix as illustrated in Fig. \ref{fNestedExample}. By Equation \eqref{eqNestedCodes}, we have $\mC_1 = \mC_0 \cup (\mC_0+V_0)$ and 
$$d(V_1, \mC_1) = \min(d(V_1, \mC_0), d(V_0+V_1, \mC_0)).$$
Next, we consider $\mC_2 = \mC_1 \cup (\mC_1 + V_1) =\mC_0 \cup (\mC_0 + V_0) \cup (\mC_0 + V_1) \cup (\mC_0 + V_0 + V_1),$
which implies \eqref{eqCosetDistance}
\begin{align*}
d(V_2, &\mC_2) = \min(d(V_2, \mC_1), d(V_1+V_2, \mC_1)) =\\
 &\min(d(V_2, \mC_0), d(V_0 + V_2, \mC_0),\\ 
 &d(V_1 + V_2, \mC_0), d(V_0+V_1+V_2, \mC_0)).
\end{align*}
\end{example}
In terms of Definition \ref{dNestedCodes}, an $l \times l$ polarization kernel $K$ can be considered as a chain $\mC(\widetilde K)$ of nested codes with $(n, 0)$ code $\mC$ and coset matrix 
$\widetilde K = 
\arraycolsep=1.15pt\def\arraystretch{0.8}
\begin{pmatrix}
K_{l-1}\\
K_{l-2}\\
\cdots\\
K_{0}
\end{pmatrix}.$ In this case the definition of coset distances is equivalent to the definition of partial distances.

 \begin{definition}
 \label{def_D_valid}
An $(n, k)$ linear code $\mC$ is referred to as $D_0^{b}$-valid, if there exist a chain $\mC(V)$ of nested codes with a coset matrix $V$ and coset distances
$$d(V_i, \mC_i) = D_i, i \in [b+1].$$
\end{definition}
A vector 
$$\mW(\mC) = (|\mW_0(\mC)|, |\mW_1(\mC)|, \cdots, |\mW_n(\mC)|),$$ where  $$\mW_i(\mC) = \set{s|\bfT_{\mC}(s) = i, s \in \F_2^{n-k}}, $$
is referred to as a \textit{weight distribution of coset leaders} (WDCL). 
\begin{proposition}
\label{prop_D_valid}
A code $\mC$ with the WDCL $\mW(\mC)$ is not $D_0^{b}$-valid if
$$
\set{D_i|i\in[b+1]}\cap \set{i|i \in[n+1],|\mW_i(\mC)| = 0} \ne \emptyset.
$$
\end{proposition}

\begin{proof}
From \eqref{eqCosetDistance} it follows  that to have $D_i = d(V_i, \mC_i)$ for any $V$ there should be at least one weight-$D_i$ coset of $\mC$.
\end{proof}

Proposition \ref{prop_D_valid} is a necessary condition for $\mC$ to be $D_0^{b}$-valid. 

Consider a code $\mC$ with the check matrix $H$ and CLWT $\bfT_{\mC}$. For a \textit{syndrome matrix} $S$ of size $(b+1)\times(n-k)$ we define a function $\bfR(S)$, which maps $S$ to a such $(b+1)\times n$ matrix $V$ as $S_i = V_i H^{\top}$, $i \in [b+1]$. Using this notation, it is convenient to rewrite \eqref{eqCosetDistance} for a particular $V = \bfR(S)$ and a chain $\mC(V)$ as
\begin{equation}
\label{eqCosetDistanceT}
d(V_i, \mC_i) = \min_{v \in \mA(V_j^{i})} d(v, \mC_j) = \min_{s \in \mA(S_j^{i})} \bfT_{\mC_j}(s).
\end{equation}

That is, the $D_0^{b}$-validity of $\mC$ can be checked by computation of $\bfT_{\mC}$ and evaluation of \eqref{eqCosetDistanceT}
for $\mC(\bfR(S))$ generated by different $S$. Obviously, we can restrict $S$ to a set  
$$\mS(\mC, D_0^{b}) = \set{S|S_i \in \mW_{D_i}(\mC), i \in [b+1]}.$$

Observe that in contrast to Proposition \ref{prop_D_valid} the equation \eqref{eqCosetDistanceT} provides an explicit method to obtain the chain of nested codes with coset distances $D_0^{b}$.

The results of this section can be naturally applied to the search of polarization kernels. Namely, given the desired PDP $\mD$, a partially constructed kernel $K_\phi^{l-1}$, which generates the kernel code $\mC = \mC_{K}^{(\phi)}$, and a CLWT $\bfT_{\mC_{K}^{(\phi)}}$, we can try to check the $D_0^{b}$-validity of $\mC$ as described above, with $D_i = \mD_{\phi-1-i}$. The case of $b = 1$ is trivial; however, in practice setting the value $b = 2, 3, 4$ turns out to be more efficient due to the complexity of CLWT recomputation. 

Unfortunately, efficient checking of $D_0^{b-1}$-validity of an $(n, k)$ linear code remains an open problem. We conjecture that this task is at least as hard as checking the existence of $(n, k, d)$ code, where $d$ is a minimum distance. 

\begin{table*}[ht]
\centering
\caption{Accuracy of different invariants}
\label{tAccInv}
\resizebox{0.97\textwidth}{!}{%
\begin{tabular}{|l|c||ccccc||ccccc|}
\hline
\multirow{2}{*}{$(n ,k, d)$} & \multirow{2}{*}{$\left |\widetilde \mL^{\star} \right|$} & \multicolumn{5}{c||}{$\left|\set{\mI(\mC)| \mC \in \widetilde\mL^{\star}}\right|$}                                                                                                          & \multicolumn{5}{c|}{$\widetilde \Delta_{\mI} = \left|\set{\mI(\mC)| \mC \in \widetilde\mL^{\star}}\right|/\left|\widetilde \mL^{\star}\right|$}                                                                                                         \\ \cline{3-12} 
                             &                    & \multicolumn{1}{c|}{$\mI_{A}$}     & \multicolumn{1}{c|}{$\mI_{\mW}$}      & 
    \multicolumn{1}{c|}{$(\mI_{A}, \mI_{\mW})$}      & \multicolumn{1}{c|}{$\mI_S$}       & $(\mI_A, \mI_\mW, \mI_S)$      & \multicolumn{1}{c|}{$\mI_{A}$}      & \multicolumn{1}{c|}{$\mI_{\mW}$}      & \multicolumn{1}{c|}{$(\mI_{A}, \mI_{\mW})$}      & \multicolumn{1}{c|}{$\mI_S$}      & $(\mI_A, \mI_\mW, \mI_S)$      \\ \hline \hline
$(18, 9, 4)$                 & 64513              & \multicolumn{1}{c|}{3258}  & \multicolumn{1}{c|}{1293}   & \multicolumn{1}{c|}{24420}  & \multicolumn{1}{c|}{64046}   & 64313   & \multicolumn{1}{c|}{0.0505} & \multicolumn{1}{c|}{0.0200} & \multicolumn{1}{c|}{0.3785} & \multicolumn{1}{c|}{0.9928} & 0.9969 \\ \hline
$(18, 10, 4)$                & 164709             & \multicolumn{1}{c|}{1688}  & \multicolumn{1}{c|}{176}    & \multicolumn{1}{c|}{11716}  & \multicolumn{1}{c|}{154115}  & 159805  & \multicolumn{1}{c|}{0.0102} & \multicolumn{1}{c|}{0.0011} & \multicolumn{1}{c|}{0.0711} & \multicolumn{1}{c|}{0.9357} & 0.9702 \\ \hline
$(18, 12, 2)$                & 73973              & \multicolumn{1}{c|}{11356} & \multicolumn{1}{c|}{50}     & \multicolumn{1}{c|}{30196}  & \multicolumn{1}{c|}{71937}   & 72422   & \multicolumn{1}{c|}{0.1535} & \multicolumn{1}{c|}{0.0007} & \multicolumn{1}{c|}{0.4082} & \multicolumn{1}{c|}{0.9725} & 0.9790 \\ \hline
$(20, 8, 6)$                 & 119156             & \multicolumn{1}{c|}{10163} & \multicolumn{1}{c|}{41165}  & \multicolumn{1}{c|}{90084}  & \multicolumn{1}{c|}{118521}  & 119109  & \multicolumn{1}{c|}{0.0853} & \multicolumn{1}{c|}{0.3455} & \multicolumn{1}{c|}{0.7560} & \multicolumn{1}{c|}{0.9947} & 0.9996 \\ \hline
$(21, 11, 4)$                & 1456430            & \multicolumn{1}{c|}{48353} & \multicolumn{1}{c|}{13420}  & \multicolumn{1}{c|}{659382} & \multicolumn{1}{c|}{1453071} & 1454715 & \multicolumn{1}{c|}{0.0332} & \multicolumn{1}{c|}{0.0092} & \multicolumn{1}{c|}{0.4527} & \multicolumn{1}{c|}{0.9977} & 0.9988 \\ \hline
$(23, 6, 8)$                 & 624629             & \multicolumn{1}{c|}{4308}  & \multicolumn{1}{c|}{623071} & \multicolumn{1}{c|}{623541} & \multicolumn{1}{c|}{619182}  & 624607  & \multicolumn{1}{c|}{0.0069} & \multicolumn{1}{c|}{0.9975} & \multicolumn{1}{c|}{0.9983} & \multicolumn{1}{c|}{0.9913} & 1      \\ \hline
$(24, 5, 8)$                 & 724530             & \multicolumn{1}{c|}{7274}  & \multicolumn{1}{c|}{723090} & \multicolumn{1}{c|}{723403} & \multicolumn{1}{c|}{713225}  & 724500  & \multicolumn{1}{c|}{0.0100} & \multicolumn{1}{c|}{0.9980} & \multicolumn{1}{c|}{0.9984} & \multicolumn{1}{c|}{0.9844} & 1      \\ \hline
$(24, 6, 8)$                 & 414734             & \multicolumn{1}{c|}{4988}  & \multicolumn{1}{c|}{414134} & \multicolumn{1}{c|}{414238} & \multicolumn{1}{c|}{412336}  & 414721  & \multicolumn{1}{c|}{0.0120} & \multicolumn{1}{c|}{0.9986} & \multicolumn{1}{c|}{0.9988} & \multicolumn{1}{c|}{0.9942} & 1      \\ \hline
$(24, 12, 6)$                & 106913             & \multicolumn{1}{c|}{2436}  & \multicolumn{1}{c|}{812}    & \multicolumn{1}{c|}{30347}  & \multicolumn{1}{c|}{106912}  & 106913  & \multicolumn{1}{c|}{0.0228} & \multicolumn{1}{c|}{0.0076} & \multicolumn{1}{c|}{0.2838} & \multicolumn{1}{c|}{1}      & 1      \\ \hline
$(26, 8, 10)$                & 534236             & \multicolumn{1}{c|}{272}   & \multicolumn{1}{c|}{216963} & \multicolumn{1}{c|}{251583} & \multicolumn{1}{c|}{504244}  & 534139  & \multicolumn{1}{c|}{0.0005} & \multicolumn{1}{c|}{0.4061} & \multicolumn{1}{c|}{0.4709} & \multicolumn{1}{c|}{0.9439} & 0.9998 \\ \hline
$(26, 11, 8)$                & 1531361            & \multicolumn{1}{c|}{5431}  & \multicolumn{1}{c|}{60327}  & \multicolumn{1}{c|}{627137} & \multicolumn{1}{c|}{1527420} & 1531272 & \multicolumn{1}{c|}{0.0035} & \multicolumn{1}{c|}{0.0394} & \multicolumn{1}{c|}{0.4095} & \multicolumn{1}{c|}{0.9974} & 0.9999 \\ \hline
$(26, 18, 4)$                & 433361             & \multicolumn{1}{c|}{17165} & \multicolumn{1}{c|}{76}     & \multicolumn{1}{c|}{74985}  & \multicolumn{1}{c|}{430508}  & 431084  & \multicolumn{1}{c|}{0.0396} & \multicolumn{1}{c|}{0.0002} & \multicolumn{1}{c|}{0.1730} & \multicolumn{1}{c|}{0.9934} & 0.9947 \\ \hline
$(28, 16, 6)$                & 92355              & \multicolumn{1}{c|}{6013}  & \multicolumn{1}{c|}{170}    & \multicolumn{1}{c|}{37558}  & \multicolumn{1}{c|}{92355}   & 92355   & \multicolumn{1}{c|}{0.0651} & \multicolumn{1}{c|}{0.0018} & \multicolumn{1}{c|}{0.4067} & \multicolumn{1}{c|}{1}      & 1      \\ \hline
\end{tabular}%
}
\end{table*}

\subsection{Computing equivalence invariants on codes}
\label{ssInvariants}
\begin{definition}
\label{defEqKernels}
Let $K$ and $\widehat K$ be $l \times l$ polarization kernels. We say that $\widehat K$ is \textbf{equivalent} to $K$, if $\widehat K = XKP$ for matrices $X$ and $P$ defined as in Proposition \ref{pAddRow}.
\end{definition}
Since the polarization properties of the equivalent kernels are the same, it is enough to consider only one kernel from each equivalence class. 

The equivalence of two polarization kernels implies the equivalence of its kernel codes. That is, we propose to prune the equivalent kernel codes $\mC_K^{(\phi)}$ from the search tree, which can be done by solving the code equivalence problem. In our case we need only to check whether two codes are equivalent, the corresponding column permutation is not used.

Unfortunately, to the best of our knowledge, there are no simple solutions of the above task, and below we provide a short review of the known methods. It was shown in \cite{petrank1997code} that the code equivalence problem can be reduced to the graph isomorphism problem in polynomial time. There is a deterministic reduction running in $O(n^\omega)$ if $\text{hull}(\mC) = \mC \cap \mC^{\perp}$ of a code $\mC$ is trivial (i.e. contains only zero codeword), where $\omega$ is a matrix multiplication exponent \cite{bardet2019permutation}. The support splitting algorithm \cite{sendrier2000finding}\ can be used to find a permutation between two equivalent codes. Algorithm \cite{leon1982computing} computes the automorphism group using a set of minimal and next to minimal weight codewords. The code equivalence test based on the canonical augmentation is proposed in \cite{bouyukliev2007about, bouyukliev2021computer}.

That is, the code equivalence problem is hard to decide. 
Hence, to keep the complexity of the polarization kernel search algorithm relatively simple, we consider the following approach.

\begin{definition}[\cite{sendrier2000finding}]
Let $\mL$ denote the set of all codes. An \textbf{invariant} $\mI$ over a set $E$ is defined as a mapping $\mL \rightarrow E$ such that any two permutation equivalent codes take the same value.
\end{definition}
Observe that by definition, in the general case the code invariant provides only sufficient condition of code non-equivalence. 
\begin{definition}
Let $\mL^{\star}$ be a set of all non-equivalent $(n, k, d)$ linear codes. We say that an invariant $\mI$ has \textbf{accuracy} 
$$
\Delta_{\mI} = {|\set{\mI(\mC)| \mC \in \mL^{\star}}|}/{|\mL^{\star}|}.
$$
\end{definition}

Given an $(n, k)$ linear code $\mC$ with the parity-check matrix $H$, we propose to compute the following several invariants. 
Let $A(\mC)$ be the weight distribution of a $\mC$, i.e. 
$$A(\mC) = (a_0, a_1, \dots, a_n),$$ where 
$$a_i = |\set{c|c \in \mC, \wt(c) = i}|.$$
The first invariant we suggest to employ is 
$$
\mI_A(\mC) = 
\begin{cases}
A(\mC), &\text{if } k \leq n/2,\\
A(\mC^{\perp}), &\text{otherwise.}
\end{cases}
$$
Since $A(\mC^{\perp})$ can be obtained from $A(\mC)$ by using the MacWilliams identities, the accuracy of both invariants is the same, but $A(\mC^{\perp})$ is easier to compute for $k > n/2$.

The second invariant we consider is $\mI_{W}(\mC) = \mW(\mC)$. It was proven \cite{britz2005covering} that the weight distribution of coset leaders $\mW(\mC)$ is not determined by $A(\mC)$, therefore, it may distinguish codes with the same weight distribution.

Following the support splitting algorithm \cite{sendrier2000finding}, we introduce the third invariant 
$$
\mI_S(\mC) = 
\begin{cases}
\set{A(s(\mC, i))|i \in [n]}, &\text{if } k \leq n/2,\\
\set{A(s(\mC^{\perp}, i))|i \in [n]}, &\text{otherwise,}
\end{cases}
$$
where $s(\mC, i)$ is a code obtained by shortening of $\mC$ on a coordinate $i$.  The weight distribution $A(s(\mC^{\perp}, i))$ can be obtained from $A(s(\mC, i))$ and $A(\mC)$ (see Lemma 18 in \cite{sendrier2000finding}), hence the invariants $A(s(\mC, i))$ and $A(s(\mC^{\perp}, i))$ have the same accuracy, but the latter is easier to compute for $k > n/2$.

% We refer two linear $(n, k)$ codes $\mC_0$ and $\mC_1$  as \textit{quasi-equivalent} %if $\mI(\mC_0) = \mI(\mC_1)$, where $\mI(\mC) = (A(\mC), W(\mC), P(\mC))$ %is a set of invariants.
Since we consider short codes, all proposed invariants can be efficiently computed, stored, hashed, and compared. Moreover, given two invariants $\mI_{F}, \mI_{G}$ we can construct a new invariant $$\mI_{F,G}(\mC) = (\mI_{F}(\mC), \mI_{G}(\mC))$$ which may have higher accuracy. To investigate this, we generated sets  $\widetilde\mL^{\star} \subset \mL^{\star}$ of some 
$(n, k, d)$ non-equivalent codes using software package \texttt{QextNewEdition} \cite{bouyukliev2021computer}, and estimated accuracy of the considered invariants for the obtained sets $\widetilde\mL^{\star}$.

Table \ref{tAccInv} presents an estimation 
$$\widetilde \Delta_{\mI}= 
{\left|\set{\mI(\mC)| \mC \in \widetilde\mL^{\star}}\right|}/
{\left|\widetilde \mL^{\star}\right|}$$
of the accuracy $\Delta_{\mI}$ of the proposed invariants. It can be seen that the invariants $\mI_{A}$ and $\mI_{\mW}$ distinguish non-equivalent codes quite weakly on their own, whereas combined in a pair, the accuracy significantly improves. At the same time, the invariant $\mI_S(\mC)$ is good enough, but the addition of $\mI_{A}$ and $\mI_{\mW}$ to $\mI_S(\mC)$ still improves the overall accuracy. 

Strictly speaking, all considered invariants do not provide the complete distinction of inequivalent codes, and consequently, if we use this modification, the algorithm is not guaranteed to find a kernel with given PDP even if it exists. Nevertheless, this approach provides significant search complexity reduction and still allows us to obtain polarization kernels, as it is demonstrated in the next sections.

%\caption{Proposed kernel search algorithm}

\subsection{Proposed search algorithm}

\begin{algorithm}[t]
\SetNoFillComment
\caption{\texttt{KernelSearch}$(K, \psi, \mD, \mI)$}
\label{alg_full}
\tcp{Input: 
\\matrix $K$  containing partially constructed kernel $K_{\phi+1}^{l-1}$;\\ current block index $\phi$;\\  the desired partial distances $\mD$;
\\ equivalence invariants $\mI$ ;
}
\If{$\psi = \mV$}{
        \Return $K$;
}
\While{\call{IsNextBlockAvailable}$(Z_\psi)$}{
$\overline V \gets$ \call{GetNextBlock}$(Z_\psi)$;\\
$K_{\phi_\psi}^{\phi_\psi +\ b_\psi-1} \gets \overline V_0^{b_\psi-1}$;\\
$H_\psi \gets$ check matrix of code $\mC_\psi=\langle K_{\phi_\psi}^{l-1}\rangle$;\\
\If{$\psi = 0$}{
Compute $\bfT_{\mC_{\psi}}$ via unexpurgated trellis;
}
\Else{
Recompute $\bfT_{\mC_{\psi}}$ via $\bfT_{\mC_{\psi-1}}$;
}
Compute $\mW(\mC_\psi)$;\\
\If{Proposition \ref{prop_D_valid} for PDs $\mD_0^{\phi_\psi+b_\psi-1}$ is false}{
\textbf{continue};\\
}
%Compute $A(\mC_\psi), P(\mC_\psi)$ and obtain $\mI(\mC_\psi)$;\\
Compute the invariant $\mI(\mC_\psi)$;\\
\If{$\mI(\mC_\psi)  \in \bI_\psi$}{
\textbf{continue};
}
\Else{
$\bI_\psi \gets \bI_\psi \cup \set{\mI(\mC_\psi)}$;
}
%$\bI_{\psi+1} \gets \emptyset$;\\
$Z_{\psi+1} \gets$ \texttt{InitializeBlock}$(\bfT_{\mC_{\psi}}, 
\mD_{\phi_\psi}^{\phi_\psi+b_\psi-1})$;\\
$\widehat K \gets $ \texttt{KernelSearch}$(K, \psi+1, \mD, \mI)$;\\
                \If{$\widehat K \ne \mathbf 0^{l \times l}$}{
                        $\bK \gets \bK \cup \set{\widehat K}$
                }
}
\Return{$\mathbf 0^{l \times l}$}\\
\tcp{Output: either kernel $K$ with the\\ required PDP $\mD$ or zero matrix $\mathbf 0^{l \times l}$}
\label{alg_fullAlg}
\end{algorithm}
In this section we assemble the methods described in Sections \ref{ssCLWT}, \ref{ssChain} and \ref{ssInvariants} into a unified algorithm.   

Suppose we are in the search tree node with an $(l, l-\phi)$ kernel code $\mC_K^{\phi}$ with a parity-check matrix $H_0$. In contrast to the basic algorithm, where descendants of dimension $l-\phi+1$ only can be obtained, now we can try to construct the descendants with $(l, l-\phi+b)$ kernel codes $\mC_K^{\phi-1-b}$, $1 \leq b \leq \phi$.  

Recall that our goal is to obtain an $l \times l$ kernel with the given PDP $\mD$. We propose to partition the PDP $\mD$ into blocks enumerated from the last PDP entry $\mD_{l-1}$. Each block is indexed by the variable $\psi$, has size $b_\psi$, starts at phase $\phi_\psi$ and ends at phase $\phi_\psi+b_\psi-1$. Let $\mV$ be a number of such blocks, then $\phi_0 +\ b_0 = l$ and $\phi_{\mV-1} = 0$. Such enumeration of PDP blocks is convenient for description, i.e. the search starts from zero block and finishes at block $\mV-1$.

Algorithm \ref{alg_full} presents the proposed search algorithm for large polarization kernels.  One should start the search by calling \textit{KernelSearch}$(\mathbf 0^{l \times l}, 0, \mD, \mI)$. We assume that the variables $\bK$ and $\bI$ are global. We also modified the basic algorithm to collect several polarization kernels  with PDP $\mD$ into the set $\bK$, which is used in Section \ref{sKernComplexity} for kernel processing complexity minimization. The set $\bI_\psi$ is a set of invariants $\mI(\mC)$, where $\mI$ can be any combination of $\mI_A, \mI_\mW,$ and $\mI_S$. In our implementation we considered the following options: $\mI_A$, $(\mI_A, \mI_\mW)$, and $(\mI_A, \mI_\mW, \mI_S)$. For efficient comparison of $\mI(C)$ universal hashing \cite{thorup2020high} is used. The block partition of $\mD$ is an input parameter of the algorithm. We recommend setting $1 \leq b_\psi \leq 4$. However, it can be adjusted for a particular PDP.

The algorithm \textit{InitializeBlock($\bfT_{\mC_{\psi}}, 
\mD_{\phi_\psi}^{\phi_\psi+b_\psi-1}$)} stores a CLWT $\bfT_{\mC_{\psi}}$ in a state variable $Z_{\psi+1}$. Partial distances $\mD_{\phi_\psi}^{\phi_\psi+b_\psi-1}$ correspond to coset distances $D_0^{b_\psi-1}$ in the following way: 
$D_{i} = \mD_{\phi_\psi+b_\psi-i-1}$, $i \in [b_\psi]$. Then, the algorithm \textit{IsNextBlockAvailable}$(Z_\psi)$ uses the stored CLWT and tries to construct the chain \eqref{eqNestedCodes} of nested codes by examining all distinct $b_\psi$-tuples $(s_0, s_1, \dots, s_{b_\psi-1})$, such that $T_{H_\psi}(s_i) = D_i$ (see Section \ref{ssChain}). For each tuple equation \eqref{eqCosetDistanceT} is used to compute the coset distance. Recall that we assume that the corresponding coset leaders $v_i$, where $s_i = v_i H_\psi^T$, are available with $\bfT_{\mC_{\psi}}$. Algorithm \textit{GetNextBlock}$(Z_\psi)$ returns the corresponding tuple $(v_0, v_1, \dots, v_{b_\psi-1})$.

\section{Maximizing the error exponent}
\label{sMaximization}

In this section we describe an application of the proposed polarization kernel search method for improvement of the lower bound $\widehat E_l$ on the maximal error exponent $E_l^{\star}$ (see Problem \ref{fExpProbl}). 
\subsection{Search for good partial distance profiles}
\label{ssBounds}

The lower and upper bounds on $E_l^{\star}$ for some $l$ were derived in \cite{korada2010polar}, \cite{presman2015binary} and \cite{lin2015linear} together with explicit kernels achieving lower bounds. For completeness, we provide some of these results in this section. 

According to \cite[Corollary 16]{korada2010polar}, achieving $E_l^{\star}$ can be done by considering only polarization kernels with nondecreasing partial distances, i.e. $\mD_\phi \leq \mD_{\phi+1}$ for $\phi \in [l-1]$.

\begin{proposition}[Minimum Distance and Partial Distance \cite{korada2010polar}]
\label{pMinPartD}
Let $\mC_0$ be a $(n,k, d_0)$ binary linear code. Let $g$ be a length-$n$ vector with $d_H(g,\mC_0) = d_2$. Let $\mC_1$ be the $(n,k+1,d_1)$ linear code obtained by adding the vector $g$ to $\mC_0$, i.e., $\mC_1 = \langle g, \mC_0 \rangle$. Then $d_1 = \min(d_0,d_2)$.
\end{proposition}

Consider an $l \times l$ kernel $K$ with PDP $\mD$. Proposition \ref{pMinPartD} implies that the minimum distance $d_K^{(\phi)}$ of the kernel codes $\mathcal C_K^{(\phi)}$, $\phi \in [l]$ is given by $\min_{\phi \leq i < l} \mD_i$. Let $d[n,k]$ be a best known minimum distance of an $(n, k)$ linear code \cite{Grassl:codetables}. Thus, we can apply the bound $d[l,l-\phi]$ to the kernel codes $\mathcal C_K^{(\phi)}$, which results in the restriction on the partial distance $\mD_\phi \leq d[l,l-\phi]$. 

\begin{lemma}(\cite[Lemma 4]{lin2015linear})
\label{lLin4}
Let $\mD_0, \mD_1, \dots, \mD_{l-1}$ be a non-decreasing partial distance profile of a kernel of size $l$. If $\mD_1 = 2$, then $\mD_i$ is even for all $i \geq 1$.
\end{lemma}

\begin{lemma}(\cite[Lemma 5]{lin2015linear})
\label{lLin5}
Let $\mD_0, \mD_1, \dots, \mD_{l-1}$ be a nondecreasing partial distance profile of a kernel $K$ of size $l$. Then for $1 \leq i \leq l$, we have
$
\sum_{i' = i}^{l}2^{l-i'}\mD_{i'-1} \leq 2^{l-i} l.
$
\end{lemma}

\begin{table}[t]
\caption{Comparison of various kernels with high error exponent}
\label{tableMaxCmp}
\centering
\scalebox{0.85}{
\begin{tabular}{|c|cccc|}
\hline
\multirow{2}{*}{$l$} & \multicolumn{4}{c|}{Error exponent} \\ \cline{2-5} & 
\multicolumn{1}{c||}{\begin{tabular}[c]{@{}c@{}}Proposed\\ method\end{tabular}} & 
\multicolumn{1}{c||}{\begin{tabular}[c]{@{}c@{}}LP upper \\bound \cite{presman2015binary}\end{tabular}} &
\multicolumn{1}{c|}{\begin{tabular}[c]{@{}c@{}}Shortened \\BCH kernels \\\cite{korada2010polar}, \cite{trofimiuk2021shortened}\end{tabular}} & 
\multicolumn{1}{c|}{\begin{tabular}[c]{@{}c@{}}Code  \\decomposition \cite{presman2015binary}\end{tabular}} \\ \hline
17 & \multicolumn{1}{c||}{\textbf{0.49361}} & \multicolumn{1}{c||}{ 0.50447} & \multicolumn{1}{c|}{0.49175} &  \\ \hline
18 & \multicolumn{1}{c||}{\textbf{0.50052}} & \multicolumn{1}{c||}{0.50925} & \multicolumn{1}{c|}{0.48968} & 0.49521 \\ \hline
19 & \multicolumn{1}{c||}{\textbf{0.50054}} & \multicolumn{1}{c||}{0.51475} & \multicolumn{1}{c|}{0.49351} & 0.49045 \\ \hline
20 & \multicolumn{1}{c||}{\textbf{0.5062}} & \multicolumn{1}{c||}{0.52190} & \multicolumn{1}{c|}{0.49659} &  \\ \hline
21 & \multicolumn{1}{c||}{\textbf{0.50868}} & \multicolumn{1}{c||}{ 0.52554} & \multicolumn{1}{c|}{0.49339} & 0.49604 \\ \hline
22 & \multicolumn{1}{c||}{\textbf{0.51181}} & \multicolumn{1}{c||}{ 0.52317} & \multicolumn{1}{c|}{0.49446} & 0.50118 \\ \hline
23 & \multicolumn{1}{c||}{\textbf{0.51612}} & \multicolumn{1}{c||}{0.52739} & \multicolumn{1}{c|}{0.50071} & 0.50705 \\ \hline
24 & \multicolumn{1}{c||}{\textbf{0.51647}} & \multicolumn{1}{c||}{0.53362} & \multicolumn{1}{c|}{0.50977} & 0.51577 \\ \hline
25 & \multicolumn{1}{c||}{\textbf{0.51683}} & \multicolumn{1}{c||}{0.53633} & \multicolumn{1}{c|}{0.50544} & 0.50608 \\ \hline
26 & \multicolumn{1}{c||}{\textbf{0.52078}} & \multicolumn{1}{c||}{} & \multicolumn{1}{c|}{0.50948} &  \\ \hline
27 & \multicolumn{1}{c||}{\textbf{0.52163}} & \multicolumn{1}{c||}{} & \multicolumn{1}{c|}{0.50969} &  \\ \hline
28 & \multicolumn{1}{c||}{\textbf{0.52197}} & \multicolumn{1}{c||}{} & \multicolumn{1}{c|}{0.51457} &  \\ \hline
29 & \multicolumn{1}{c||}{\textbf{0.52109}} & \multicolumn{1}{c||}{} & \multicolumn{1}{c|}{0.5171} &  \\ \hline
\end{tabular}
}
\end{table}

\begin{table*}[ht]
\caption{Polarization parameters of the proposed kernels with the best error exponent}
\label{tableMax}
\centering
%\footnotesize
\scalebox{0.85}{
\begin{tabular}{|c|c|c|l|cc|}
\hline
\multirow{2}{*}{$l$} & \multirow{2}{*}{$E$} & \multirow{2}{*}{$\mu$} & \multicolumn{1}{c|}{\multirow{2}{*}{Partial distance profile}} & \multicolumn{2}{c|}{Processing complexity} \\ \cline{5-6} 
 &  &  & \multicolumn{1}{c|}{} & \multicolumn{1}{c|}{RTPA} & Viterbi \\ \hline
17 & 0.493607 & 3.461 & 1, 1, 2, 2, 2, 3, 4, 4, 4, 5, 6, 7, 8, 8, 8, 8, 16 & \multicolumn{1}{c|}{1073} & 8911 \\ \hline
18 & 0.50052 & 3.474 & 1, 2, 2, 2, 2, 2, 4, 4, 4, 6, 6, 6, 6, 8, 8, 10, 10, 12 & \multicolumn{1}{c|}{2956} & 15324 \\ \hline
19 & 0.50054 & 3.417 & 1, 2, 2, 2, 2, 2, 4, 4, 4, 4, 6, 6, 6, 8, 8, 8, 10, 10, 16 & \multicolumn{1}{c|}{4020} & 20097 \\ \hline
20 & 0.5062 & 3.414 & 1, 2, 2, 2, 2, 2, 4, 4, 4, 4, 6, 6, 8, 8, 8, 8, 8, 8, 12, 16 & \multicolumn{1}{c|}{3546} & 35843 \\ \hline
21 & 0.50868 & 3.347 & 1, 2, 2, 2, 2, 2, 4, 4, 4, 4, 6, 6, 6, 6, 8, 8, 10, 10, 10, 14, 14 & \multicolumn{1}{c|}{6500} & 39566 \\ \hline
22 & 0.51181 & 3.331 & 1, 2, 2, 2, 2, 2, 4, 4, 4, 4, 6, 6, 6, 6, 8, 8, 8, 10, 10, 10, 12, 20 & \multicolumn{1}{c|}{9668} & 47875 \\ \hline
23 & 0.51612 & 3.318 & 1, 2, 2, 2, 2, 2, 4, 4, 4, 4, 6, 6, 6, 6, 8, 8, 8, 10, 10, 10, 12, 14, 16 & \multicolumn{1}{c|}{13996} & 92963 \\ \hline
24 & 0.51647 & 3.287 & 1, 2, 2, 2, 2, 2, 4, 4, 4, 4, 4, 6, 6, 6, 8, 8, 8, 8, 10, 12, 12, 12, 16, 16 & \multicolumn{1}{c|}{18262} & 103885 \\ \hline
25 & 0.51683 & 3.281 & 1, 2, 2, 2, 2, 2, 4, 4, 4, 4, 4, 6, 6, 6, 8, 8, 8, 8, 8, 10, 12, 12, 12, 16, 18 & \multicolumn{1}{c|}{22690} & 120458 \\ \hline
26 & 0.52078 & 3.248 & 1, 2, 2, 2, 2, 2, 4, 4, 4, 4, 4, 6, 6, 6, 6, 8, 8, 8, 10, 10, 12, 12, 12, 12, 16, 20 & \multicolumn{1}{c|}{37820} & 222309 \\ \hline
27 & 0.52163 & 3.3 & 1, 2, 2, 2, 2, 2, 4, 4, 4, 4, 4, 6, 6, 6, 6, 8, 8, 8, 10, 10, 10, 12, 12, 12, 12, 16, 20 & \multicolumn{1}{c|}{48756} & 242714 \\ \hline
28 & 0.52197 & 3.227 & 1, 2, 2, 2, 2, 2, 4, 4, 4, 4, 4, 6, 6, 6, 6, 6, 8, 8, 8, 10, 10, 10, 12, 12, 14, 14, 16, 24 & \multicolumn{1}{c|}{63826} & 316029 \\ \hline
29 & 0.52109 & 3.23 & 1, 2, 2, 2, 2, 2, 4, 4, 4, 4, 4, 4, 6, 6, 6, 6, 8, 8, 8, 10, 10, 10, 12, 12, 14, 14, 16, 16, 20 & \multicolumn{1}{c|}{91552} & 449328 \\ \hline
\end{tabular}
}
\end{table*}

We also describe the linear programming (LP) bounds on partial distances \cite{presman2015binary}.  For fixed positive integers $n$ and $\lambda$, the \textit{Krawtchouk polynomial} $P_{k}(x)$ of degree $k$ is defined by
\begin{equation}
P_{k}(x)=\sum_{j=0}^{k}(-1)^{j} \lambda^{k-j}\left(\begin{array}{l}
x \\
j
\end{array}\right)
\left(
\begin{array}{c}
n-x \\
k-j
\end{array}\right), \quad 0 \leq k \leq n
\end{equation}
For a linear $(n, k)$ code $\mC$ over $\mathbb F^{q}, q = \lambda + 1,$ with weight distribution $A$ we have \cite{delsarte1972bounds}
\begin{equation}
\sum_{i = 0}^{n} A_{i} P_{k}(i) \geqslant -\lambda^k {n \choose k}, \quad k \in n
\end{equation}
For a binary $l \times l$ kernel $K$ with PDP $\mD$ the following expressions hold \cite{presman2015binary}
\begin{align}
\label{fLPKernel1} 
&A_{K}^{(\phi)}[i] \geq 0, &0 \leq i, \phi \leq l,\\
&\sum_{i = 0}^{l} A_{K}^{(\phi)}[i] = 2^{(l-\phi)}, &0 \leq \phi \leq l,\\
&A_{K}^{(\phi-1)}[i] - A_{K}^{(\phi)}[i] = 0, &0 \leq i < \mD_{\phi}, 1 \leq \phi \leq l,\\
&A_{K}^{(\phi-1)}[i] - A_{K}^{(\phi)}[i] \geq 0, &\mD_{\phi} \leq i \leq n, 1 \leq \phi \leq l,\\
 \label{fLPKernel2}
&\sum_{i = 0}^{l} A_k^{(\phi)}[i] P_{k}(i) \geq - {l \choose k}, &k \in l
\end{align}

An arbitrary PDP $\mD$ can be tested for \eqref{fLPKernel1}--\eqref{fLPKernel2} via the integer or constraint programming solver. In this work we used Google OR-Tools for C++ \cite{ortools}. 

That is, to improve the lower bound $\widehat E_l$ on $E^{\star}_l$, we consider the approach similar to the one described in \cite{lin2015linear}. Namely,

\begin{enumerate}
\item Generate a set $\widetilde{\mathbf D}_l$ of LP-valid nondecreasing PDPs $\mD$ with $\mD_{\phi}  \leq d[l,l-\phi]$ satisfying Lemmas \ref{lLin4} and \ref{lLin5}, and with $\frac{1}{l} \sum_{i=0}^{l-1} \log_l \mD_i$ larger than the existing lower bounds on $E_l$;
\item  Run Algorithm \ref{alg_full} for all $\mD \in \widetilde{\mathbf D}_l$.
\end{enumerate}

Due to the tremendous search space, for some PDP $\mD$ the algorithm may not produce a kernel for a long time. In this case, it should be terminated. Nevertheless, during the computer search we observed that if a kernel with given PDP exists, Algorithm \ref{alg_full} often returns it quickly (in a few seconds). 
%This is the advantage of depth-first search, which allows us to perform %maximization of $E_l$.

\subsection{Results}

Table \ref{tableMaxCmp} presents the error exponent of the kernels obtained by the proposed approach. The obtained kernels can be found in Appendix. It can be seen that for all considered sizes, the obtained kernels have better error exponent in comparison with the shortened BCH kernels, introduced in \cite{korada2010polar} and refined in \cite{trofimiuk2021shortened}, and the kernels, obtained from the code decomposition \cite{presman2015binary}. We also provide the upper bounds on $E^{\star}_l$ obtained by linear programming \cite{presman2015binary}. Observe that the gap between the error exponent of the proposed kernels and the upper bound of the LP is no greater than $0.02$, which indicates the efficiency of the proposed approach.

Table \ref{tableMax} reports the scaling exponent, PDP and processing complexity (measured as a number of addition and comparison operations) of the obtained kernels for Viterbi-based processing algorithm \cite{griesser2002aposteriori} and the recently proposed recursive trellis processing algorithm (RTPA) \cite{trifonov2021recursive}. It can be seen that the complexity of RTPA is significantly lower compared with Viterbi processing. For each kernel we also tried several million random column permutations to reduce its RTPA processing complexity. Nevertheless, we conjecture that the processing complexity of the obtained kernels can be further decreased by finding the appropriate column permutation.  

We did not try to find kernels of the considered sizes with the lowest scaling exponent. We suppose that it can be further decreased.

\section{Kernels with low processing complexity}
\label{sKernComplexity}
\begin{table*}[ht]
\caption{Polarization parameters of kernels which admit low complexity processing}
\label{tableCompl}
\centering
\scalebox{0.85}{
\begin{tabular}{
|c%|@{\hspace{1mm}}>{\centering}p{2.3mm}@{\hspace{1mm}}
|c%|@{\hspace{1mm}}>{\centering}p{10mm}@{\hspace{1mm}}
|c%|@{\hspace{1mm}}>{\centering}p{9mm}@{\hspace{1mm}}
|l|c|c|}
\hline
\multirow{2}{*}{$l$} & \multirow{2}{*}{$E$} & \multirow{2}{*}{$\mu$} & \multicolumn{1}{c|}{\multirow{2}{*}{Partial distance profile}}                                                     & \multicolumn{2}{c|}{Processing complexity} \\ \cline{5-6} 
                     &                      &                        & \multicolumn{1}{c|}{}                                                                                       & RTPA         & Viterbi        \\ \hline
18                   & 0.49521              & 3.648                & 1, 2, 2, 4, 2, 2, 2,   4, 4, 6, 4, 6, 8, 8, 8, 8, 8, 16                                                     & 1000                      & 6506           \\ \hline
20                   & 0.49943              & 3.649                & 1, 2, 2, 4, 2, 4, 2,   2, 4, 4, 6, 8, 8, 8, 4, 8, 12, 8, 8, 16                                              & 478                       & 5622           \\ \hline
24                   & 0.50291              & 3.619                & 1, 2, 2, 4, 2, 4, 2,   4, 6, 2, 4, 4, 8, 8, 12, 4, 4, 8, 8, 12, 12, 8, 16, 16                               & 365                       & 6319           \\ \hline
24                   & 0.51577              & 3.311                & 1, 2, 2, 2, 2, 2, 4, 4, 4, 4, 4, 4, 8, 8, 8, 8, 8, 8, 8, 12, 12, 12, 16, 16                               & 2835                       & 29782           \\ \hline
27                   & 0.49720              & 3.766                & 1, 2, 2, 4, 2, 2, 4,   4, 6, 2, 4, 4, 6, 6, 8, 8, 10, 12, 4, 4, 8, 8, 12, 12, 8, 16, 16                     & 998                       & 14907          \\ \hline
32                   & 0.52194              & 3.421                
& 1,\! 2,\! 2,\! 4,\! 2,\! 4,\! 2,\! 4,\! 6,\! 8,\! 2,\! 4,\! 6,\! 8,\! 4, 6,\! 8,\! 12,\! 4,\! 8,\! 12,\! 16,\! 4,\! 8,\! 12,\! 16,\! 8,\! 16,\! 
8,\! 16,\! 16,\! 32 & 526                       & 13608          \\ \hline
\end{tabular}
}
\end{table*}

Although the kernels from Table \ref{tableMax} have excellent polarization parameters, their practical usage is limited due to the tremendous processing complexity. We propose to consider polarization kernels with degraded polarization parameters and non-monotonic partial distances. This is motivated by the fact that the trellis state complexity of an $(n, k, d)$ linear code is lower bounded \cite{muder1988minimal}. That is, by degrading the error exponent and permuting the partial distances, we reduce the minimal distance of the corresponding kernel codes, thus reducing the lower bound on trellis state complexity. For example, such kernels of size 16 and 32 were proposed in \cite{trofimiuk2019construction32} and \cite{trofimiuk2021window}.

We propose to use the algorithm given in Section \ref{ssBounds} to obtain several PDPs corresponding to a reduced error exponent. Then we permute some PDP entries and try to construct the corresponding kernels using Alg. \ref{alg_full}.

The results of complexity minimization are presented in Table \ref{tableCompl}. The proposed kernels can be found in Appendix. 

\section{Numeric results}
\label{sNumberic}

Several mixed-kernel (MK) polar subcodes with obtained kernels were constructed. Their performance was investigated for the case of AWGN\ channel with BPSK modulation. The sets of frozen symbols were obtained by the method proposed in \cite{trifonov2019construction}. For MK polar subcodes we used different order of kernels and present the results for the best one.

\begin{figure}[ht]
\centering
\begin{subfigure}[b]{0.4\textwidth}
\centering
\resizebox{!}{0.67\linewidth}{
\input{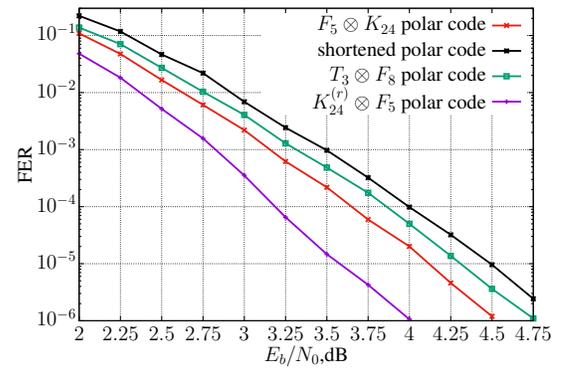}}
\caption{$(768, 384)$ polar codes}
\label{fSC768}
\end{subfigure}

%\hspace{-3mm}
\begin{subfigure}[b]{0.4\textwidth}
\resizebox{!}{0.67\linewidth}{
\input{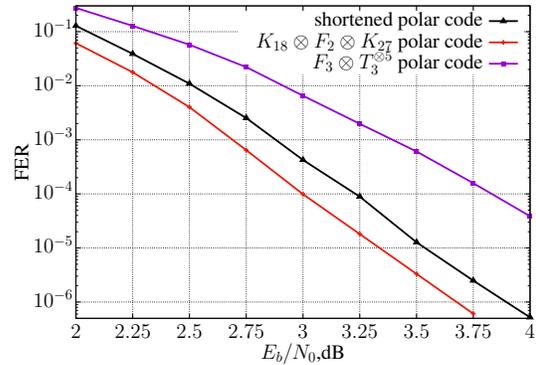}}
\caption{$(1944, 972)$ polar codes}
\label{fSC1944}
\end{subfigure}
\caption{Successive cancellation decoding performance}
\label{fSCAll}
\end{figure}

\begin{figure}[ht]
\centering
\resizebox{!}{0.7\linewidth}{
\input{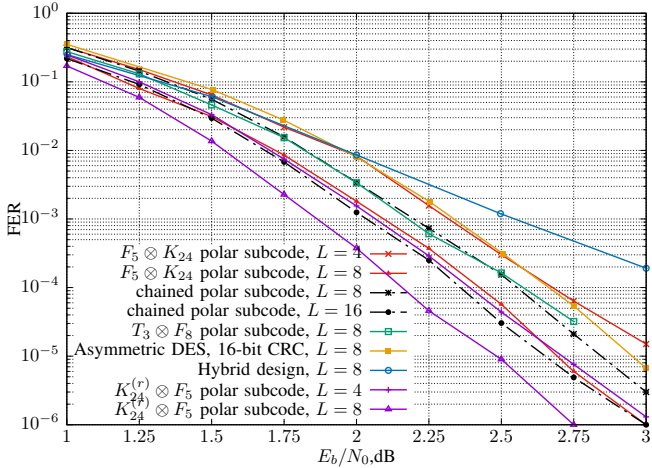}}
\caption{Performance of $(768, 384)$ polar subcodes under list decoding}
\label{f768Search}
\end{figure}

Figure \ref{fSC768} presents the performance of different MK polar codes under successive cancellation (SC) decoding. We report the results for the randomized MK polar codes with the proposed $K_{24}$ and $K_{24}^{(r)}$ kernels, $2^t \times 2^t$ Ar{\i}kan matrix $
F_t = \left(
\arraycolsep=1.15pt\def\arraystretch{0.5}
\begin{array}{cc}
1&0\\
1&1
\end{array} \right)^{\otimes t}
$
 and $3\times 3$ kernel $T_3$ introduced in \cite{bioglio2020multikernel}. For comparison we also added the performance of the shortened polar code. It can be observed that polar codes with the proposed large polarization kernels indeed provide performance gain due to improved polarization properties. 

Figure \ref{f768Search} illustrates the performance of different $(768, 384)$ polar codes decoded by the SC list (SCL) algorithm \cite{tal2015list} with list size $L$. For list decoding we use the randomized MK polar subcodes \cite{trifonov2017randomized}, \cite{trifonov2019construction}. We also include the results for the descending (DES) asymmetric construction with 16-bit CRC \cite{cavatassi2019asymmetric} and the improved hybrid design \cite{bioglio2019improved} of the MK polar codes together with the performance of the randomized chain polar subcodes \cite{trifonov2018randomized}. It can be seen that polar codes with $F_5 \otimes K_{24}$ polarizing transform under SCL decoding with $L = 8$ outperform the chained and $T_3 \otimes F_8$ polar subcodes, asymmetric DES, and hybrid construction under SCL with the same list size. Moreover, it provides almost the same performance as chained polar subcode decoded with $L = 16$. Observe that despite of having degraded polarization parameters (compared with $K_{24}^{(r)}$ kernel), kernel $K_{24}$ still provides a noticeable performance gain compared with other code constructions.

%Moreover, $K_{24}$ admits very low complexity processing by RTPA. This demonstrates %the importance of finding polarization kernels with low processing complexity.

It can be seen that due to the excellent polarization properties of $K_{24}^{(r)}$, polar subcode with $K_{24}^{(r)} \otimes  F_5$ under SCL with $L = 4$ provides the same performance as $F_5 \otimes K_{24}$ with $L = 8$. 
\begin{figure}[ht]
\centering
\resizebox{!}{0.7\linewidth}{\input{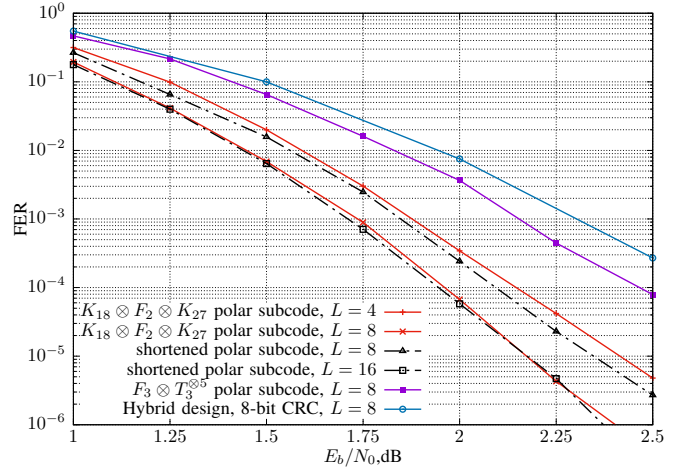}}
\caption{Performance of $(1944, 972)$ polar codes}
\label{f1944}
\end{figure}

Figure \ref{f1944} presents the performance of various $(1944, 972)$ polar codes. It can be seen that the polar subcode with $K_{18} \otimes F_{2} \otimes K_{27}$ polarizing transform under SCL decoding with $L = 4$ provides significant performance gain compared to the MK polar subcode with $F_3 \otimes T_3^{\otimes 5}$ polarizing transform and hybrid design with 8-bit CRC \cite{bioglio2020multikernel}. Moreover,  $K_{18} \otimes F_{2} \otimes K_{27}$ polar subcode with $L = 8$ has the same performance as shortened polar subcode decoded with $L = 16$. We should also point out that despite $E(K_{18}) <\ 0.5$ and $E(K_{27}) <\ 0.5$, the resulting code achieves a significant performance gain compared with shortened polar code with $F_{11}$ having $E(F_{11}) = 0.5$.  It follows that kernels with sizes $l \ne 2^t$ and even $E(K) <\ 0.5$ should be studied and can be considered in mixed kernel polar codes, since polar codes with such kernels outperform Ar{\i}kan-based polar code constructions.

It can be observed on Figure \ref{fSC1944} that $K_{18} \otimes F_{2} \otimes K_{27}$ polar code also outperforms other code constructions under SC\ decoding.

\section{Conclusion and open problems}
In this paper a novel search method for large polarization kernels with given partial distance profile was proposed. This algorithm is based on depth-first search over sets of candidate rows. The proposed algorithm was used to improve the lower bounds on the maximal error exponent of sizes $17 \leq l \leq 29$. 

We also demonstrated application of the proposed method for finding kernels which admit low complexity processing by the recursive trellis processing algorithm. Numerical results show the advantage of the polar codes with the obtained kernels compared with shortened polar codes with the Ar{\i}kan kernel and polar codes with small kernels.

In the following, we outline multiple open problems with large kernels:
\begin{itemize}
\item \textit{Tight bounds on partial distances.} There is a significant gap between the lower and upper bounds of the maximum error exponent of $l \times l$ kernels for large $l$. Moreover, we have almost no bounds for non-monotonic partial distances, except the LP bound.

\item The problem of \textit{scaling exponent minimization} remains unsolved. It is unknown whether the proposed method can be extended to explicitly solve it.
\item \textit{Processing complexity.} For some PDP the proposed algorithm can produce kernels which admit low complexity processing by RTPA. However, it is unknown how to obtain kernel with given PDP and minimum processing complexity.
\end{itemize}

\section*{Acknowledgments}
The author thanks the editor and the reviewers for their beneficial
 comments.

\appendix
\section{Proposed polarization kernels}
\label{appendixKern}
%\addcontentsline{toc}{chapter}{Appendix A. Proposed polarization kernels}

%\chapter*{Appendix 1}
%\appendixdesc{Акты внедрения}
%\addcontentsline{toc}{chapter}{Appendix 1}

Figures \ref{fBestKernels17}, \ref{fBestKernels21} and \ref{fBestKernels27} presents  polarization kernels with the best error exponent, obtained in Section \ref{sMaximization}. Figure \ref{fSimpleKernels} presents  polarization kernels which admit low processing complexity, obtained in Section \ref{sKernComplexity}.

\begin{figure*}[ht]
\centering
\scalebox{0.8}{
%\parbox{1.7\textwidth}
{ 
$
\arraycolsep=1.2pt\def\arraystretch{0.6}
\begin{array}{ccc}
{
K_{17}^{\ast} = \left(\begin{array}{ccccccccccccccccc}
0&1&0&0&0&0&0&0&0&0&0&0&0&0&0&0&0\\
1&0&0&0&0&0&0&0&0&0&0&0&0&0&0&0&0\\
0&0&1&1&0&0&0&0&0&0&0&0&0&0&0&0&0\\
0&0&0&0&0&1&0&0&0&0&0&1&0&0&0&0&0\\
0&1&1&0&0&0&0&0&0&0&0&0&0&0&0&0&0\\
1&1&0&0&0&0&1&0&0&0&0&0&0&0&0&0&0\\
0&0&0&1&1&1&1&0&0&0&0&0&0&0&0&0&0\\
0&0&0&0&0&1&1&1&1&0&0&0&0&0&0&0&0\\
0&0&0&0&0&0&0&0&0&1&1&0&0&0&0&1&1\\
1&0&0&0&0&1&0&0&1&0&0&0&1&0&0&0&1\\
0&0&0&1&1&1&0&0&1&1&1&0&0&0&0&0&0\\
1&0&0&1&0&1&1&1&0&1&0&0&1&0&0&0&0\\
0&0&0&1&1&1&1&0&0&0&1&0&1&0&1&1&0\\
0&0&0&0&0&1&1&1&1&0&0&1&1&0&0&1&1\\
0&0&1&1&0&1&0&1&0&1&1&0&0&0&0&1&1\\
0&1&0&0&1&0&1&0&1&1&1&0&0&0&0&1&1\\
0&1&1&1&1&1&1&1&1&1&1&1&1&1&1&1&1\\
\end{array}\right)
}
&
{
K_{18}^{\ast} = \left(\begin{array}{cccccccccccccccccc}
0&0&0&0&0&1&0&0&0&0&0&0&0&0&0&0&0&0\\
0&1&0&0&0&0&0&0&0&0&0&0&0&0&0&1&0&0\\
0&0&1&0&0&0&0&0&0&0&0&0&1&0&0&0&0&0\\
0&0&0&0&0&0&0&0&0&0&0&0&1&0&0&0&0&1\\
0&0&0&0&0&1&0&0&0&0&0&0&0&0&0&0&1&0\\
0&0&0&0&0&0&0&0&0&0&0&0&0&0&0&0&1&1\\
1&0&0&0&1&0&1&0&0&0&1&0&0&0&0&0&0&0\\
1&0&0&1&0&0&1&0&0&0&0&0&0&1&0&0&0&0\\
0&0&0&0&0&0&1&0&0&0&0&1&0&1&0&0&1&0\\
0&1&0&0&1&0&1&0&1&0&1&0&0&1&0&0&0&0\\
1&1&0&0&0&0&0&0&0&0&1&0&0&1&1&0&1&0\\
1&0&0&1&0&0&0&1&1&0&1&0&0&1&0&0&0&0\\
0&0&0&1&0&1&1&1&0&0&0&1&0&0&0&0&1&0\\
1&1&1&0&0&1&0&0&0&0&1&0&0&0&0&1&1&1\\
0&1&0&1&0&0&0&0&1&0&1&0&0&1&1&1&0&1\\
0&1&1&1&0&1&1&1&0&1&1&0&0&1&0&0&1&0\\
0&0&0&0&1&1&1&0&0&0&1&0&1&1&1&1&1&1\\
1&1&0&1&1&0&1&1&1&0&0&1&0&1&1&1&1&0\\
\end{array}\right)
} 
&
{
K_{19}^{\ast} = \left(\begin{array}{ccccccccccccccccccc}
0&0&0&0&0&0&0&0&0&0&0&0&0&0&1&0&0&0&0\\
0&0&0&0&0&0&0&0&0&0&0&0&1&0&0&0&0&0&1\\
1&0&0&0&0&0&1&0&0&0&0&0&0&0&0&0&0&0&0\\
0&0&0&0&0&0&1&1&0&0&0&0&0&0&0&0&0&0&0\\
0&0&0&0&0&0&0&0&0&0&0&0&0&0&0&0&1&0&1\\
0&0&0&0&0&0&0&0&0&0&0&0&0&0&0&0&0&1&1\\
1&0&0&0&0&1&0&0&0&0&0&0&1&0&0&0&0&0&1\\
0&0&1&0&0&0&0&0&1&0&0&0&1&1&0&0&0&0&0\\
0&0&0&0&0&0&0&0&0&0&0&1&1&0&0&0&1&1&0\\
0&0&0&0&0&0&0&0&0&0&0&0&1&0&1&0&1&0&1\\
0&0&1&0&1&0&0&0&1&0&1&0&1&0&0&0&1&0&0\\
0&0&1&0&0&0&0&0&1&1&1&1&0&1&0&0&0&0&0\\
0&0&0&1&1&0&0&0&1&1&0&1&0&0&1&0&0&0&0\\
0&0&1&0&0&0&0&1&1&0&0&1&1&1&0&0&1&0&1\\
0&0&1&1&0&0&0&0&0&1&1&0&0&1&0&1&1&1&0\\
0&1&1&0&0&0&0&0&1&0&1&0&0&1&1&1&0&0&1\\
0&1&1&1&0&1&1&1&0&0&0&1&0&1&0&1&0&0&1\\
1&1&0&0&1&0&1&1&1&1&1&1&0&1&0&0&0&0&0\\
0&1&1&1&1&0&1&1&1&1&1&0&1&1&1&1&1&1&1\\
\end{array}\right)
}
\\
{
K_{20}^{\ast} = \left(\begin{array}{cccccccccccccccccccc}
0&0&0&0&0&0&1&0&0&0&0&0&0&0&0&0&0&0&0&0\\
0&0&0&0&0&1&0&1&0&0&0&0&0&0&0&0&0&0&0&0\\
0&0&0&0&1&0&0&0&0&1&0&0&0&0&0&0&0&0&0&0\\
0&0&0&0&0&0&0&1&0&1&0&0&0&0&0&0&0&0&0&0\\
0&1&0&0&0&0&0&0&0&0&0&0&0&0&0&0&0&0&0&1\\
0&0&0&0&0&0&1&0&0&0&0&0&0&0&0&0&0&0&0&1\\
0&0&0&0&0&0&1&0&1&1&0&0&0&0&0&0&0&0&1&0\\
0&0&0&0&0&1&0&0&0&0&0&0&0&1&1&0&0&1&0&0\\
0&1&0&0&0&0&1&0&0&0&0&0&0&0&0&0&1&0&0&1\\
0&0&0&0&1&1&1&1&0&0&0&0&0&0&0&0&0&0&0&0\\
0&0&0&0&1&1&0&1&0&0&0&0&1&0&0&1&0&1&0&0\\
0&0&1&0&0&0&0&1&1&0&0&0&0&1&0&0&1&0&1&0\\
0&0&0&0&1&1&0&0&1&1&0&1&0&1&0&0&1&0&1&0\\
0&0&1&0&1&1&0&1&0&0&1&0&0&1&0&0&1&1&0&0\\
1&0&1&0&0&0&0&0&0&1&1&0&0&1&0&1&0&1&0&1\\
0&1&1&0&0&1&1&0&0&0&1&0&0&1&1&1&0&0&0&0\\
0&1&0&0&0&1&0&0&0&1&0&0&1&1&0&1&1&1&0&0\\
0&0&1&0&1&1&0&1&1&0&0&0&1&0&0&0&0&0&1&1\\
1&1&1&1&0&0&0&1&1&1&1&0&0&1&1&1&0&0&1&0\\
1&1&1&1&1&1&1&0&1&1&1&1&1&0&1&0&1&1&0&1\\
\end{array}\right)
}
&
{
K_{21}^{\ast} = \left(\begin{array}{ccccccccccccccccccccc}
0&0&0&0&0&0&0&0&0&0&0&0&0&0&0&0&1&0&0&0&0\\
0&0&0&0&0&0&0&0&0&0&0&0&0&0&0&1&0&0&1&0&0\\
0&0&0&1&0&0&0&0&0&0&0&0&0&0&0&1&0&0&0&0&0\\
0&0&0&0&0&0&0&0&1&0&0&0&0&0&0&0&0&1&0&0&0\\
0&0&0&0&0&0&0&0&1&1&0&0&0&0&0&0&0&0&0&0&0\\
0&0&1&0&0&0&0&0&0&0&0&0&0&0&0&0&0&0&0&1&0\\
0&0&1&1&0&0&0&0&1&0&0&0&0&1&0&0&0&0&0&0&0\\
0&1&0&0&1&0&0&1&0&0&0&0&0&0&0&0&0&0&0&1&0\\
0&0&1&0&1&0&0&0&0&1&0&0&0&1&0&0&0&0&0&0&0\\
0&0&0&0&1&0&0&0&0&0&0&0&0&1&0&0&1&0&0&1&0\\
0&0&0&1&0&1&0&0&0&0&1&0&0&1&1&0&0&0&1&0&0\\
0&0&0&0&0&1&0&1&0&0&1&0&1&0&0&1&0&0&0&1&0\\
0&0&0&1&0&0&0&0&1&0&1&0&0&0&1&1&0&0&0&1&0\\
0&1&1&1&0&0&0&0&0&0&1&0&0&0&0&1&1&0&0&0&0\\
0&0&1&1&0&0&0&1&1&0&1&0&1&0&0&1&0&0&0&0&1\\
0&0&0&1&0&0&0&0&0&0&0&1&1&0&0&1&1&0&1&1&1\\
1&1&1&1&0&1&0&0&1&0&0&0&1&0&0&1&0&0&1&1&0\\
0&1&0&1&0&0&0&1&1&0&0&1&0&0&1&0&1&1&1&0&1\\
0&0&1&0&1&1&0&0&0&1&1&1&1&0&0&0&1&0&1&0&1\\
1&0&1&1&0&0&1&1&1&1&1&1&1&0&0&1&1&1&1&0&0\\
0&1&0&1&1&1&0&1&0&0&1&1&1&1&1&1&0&0&1&1&1\\
\end{array}\right)
} 
&
{
K_{22}^{\ast} = \left(\begin{array}{cccccccccccccccccccccc}
0&0&0&0&0&0&0&0&0&0&0&0&1&0&0&0&0&0&0&0&0&0\\
0&0&0&0&0&0&0&0&0&1&0&0&0&0&1&0&0&0&0&0&0&0\\
0&0&1&0&0&0&0&0&0&0&0&0&0&0&1&0&0&0&0&0&0&0\\
0&0&0&0&0&0&0&0&0&1&0&0&0&0&0&0&0&0&0&0&1&0\\
0&1&0&0&0&0&0&0&0&0&0&0&0&0&0&1&0&0&0&0&0&0\\
1&0&0&0&0&0&0&0&0&0&0&0&1&0&0&0&0&0&0&0&0&0\\
0&0&1&0&0&0&0&0&0&0&0&0&1&0&0&1&0&0&0&1&0&0\\
0&0&0&1&0&0&0&0&0&0&1&0&0&0&0&0&1&0&0&0&1&0\\
0&0&0&0&0&1&0&0&0&1&0&0&1&0&0&0&0&1&0&0&0&0\\
0&0&0&0&1&1&0&0&0&0&0&0&1&0&0&0&0&0&0&0&0&1\\
1&0&0&0&0&0&1&1&0&1&0&0&0&1&0&0&1&0&0&0&0&0\\
0&0&0&0&0&0&0&0&0&0&0&0&1&1&1&0&1&1&0&1&0&0\\
1&0&0&0&0&0&1&0&0&0&0&0&1&0&0&0&1&0&0&1&0&1\\
1&0&0&1&0&1&0&0&0&0&0&0&1&0&1&0&1&0&0&0&0&0\\
1&1&0&0&0&0&0&0&0&1&0&1&1&1&0&0&1&0&0&1&0&0\\
1&0&1&0&0&0&1&0&0&0&1&1&0&0&0&0&0&1&0&1&0&1\\
0&0&1&0&1&0&0&1&0&0&1&1&0&0&0&0&1&0&1&0&0&1\\
1&0&1&0&1&0&0&0&1&0&1&1&1&0&1&0&1&0&0&1&0&0\\
0&0&1&0&1&1&1&0&0&0&0&1&0&1&1&0&0&0&1&1&0&1\\
0&1&1&1&0&0&0&0&0&0&1&1&0&1&1&0&1&1&1&0&0&0\\
0&1&1&1&0&0&1&1&1&0&1&0&0&0&1&1&0&0&1&1&0&1\\
1&1&1&1&1&1&1&1&0&1&1&1&1&1&1&0&1&1&1&1&1&1\\
\end{array}\right)
} 
\\
\end{array}
$
}}
\caption[]{Polarization kernels with the best error exponent of sizes 17-22}
\label{fBestKernels17}
\end{figure*}

\begin{figure*}[ht]
\centering
\scalebox{0.8}{
%\parbox{1.7\textwidth}
{ 
$
\arraycolsep=1.2pt\def\arraystretch{0.6}
\begin{array}{cc}

{
K_{23}^{\ast} = \left(\begin{array}{ccccccccccccccccccccccc}
0&0&0&0&0&0&0&0&0&0&0&0&0&0&1&0&0&0&0&0&0&0&0\\
0&0&0&0&1&0&0&0&0&0&0&0&0&0&0&0&0&1&0&0&0&0&0\\
0&0&0&0&0&0&0&1&0&0&0&0&0&0&0&0&0&1&0&0&0&0&0\\
0&0&0&0&0&0&0&0&0&0&0&0&0&1&0&1&0&0&0&0&0&0&0\\
0&0&0&0&0&0&0&0&0&0&0&0&0&0&1&1&0&0&0&0&0&0&0\\
0&0&0&0&0&0&0&0&0&0&0&1&0&0&1&0&0&0&0&0&0&0&0\\
0&0&0&0&1&0&0&0&0&0&1&0&0&0&0&0&1&0&0&0&0&0&1\\
0&0&0&0&0&0&1&0&0&1&0&0&0&0&0&0&1&0&0&0&0&1&0\\
0&0&0&0&0&0&0&0&0&0&0&1&1&0&0&1&0&0&0&0&0&1&0\\
0&0&0&0&1&0&0&1&0&0&0&1&1&0&0&0&0&0&0&0&0&0&0\\
0&0&0&0&1&1&0&0&0&0&0&1&1&0&0&0&0&0&0&1&0&0&1\\
0&1&0&1&0&0&0&0&0&0&0&1&0&0&1&0&0&1&0&0&0&0&1\\
0&0&0&1&0&0&1&1&1&0&0&0&1&0&0&0&0&0&0&0&0&0&1\\
1&1&1&0&0&1&0&0&0&0&0&0&0&0&0&1&0&0&0&0&0&0&1\\
0&0&0&1&1&1&0&0&0&0&1&0&0&0&0&0&1&1&1&0&1&0&0\\
0&1&0&0&0&1&0&0&0&0&0&0&0&0&1&0&1&1&0&1&1&1&0\\
0&0&0&0&1&1&0&0&1&0&1&0&1&0&0&0&0&0&0&1&1&1&0\\
1&1&0&1&1&0&0&0&0&0&0&0&1&0&0&0&1&1&0&1&1&0&1\\
0&1&1&1&0&0&1&1&1&0&0&0&0&0&0&1&1&0&0&1&0&1&0\\
0&0&1&1&1&1&1&1&0&0&0&1&0&0&0&0&1&1&0&0&0&0&1\\
1&1&1&0&0&1&1&0&0&1&0&0&0&1&0&0&1&1&1&1&0&0&1\\
1&1&1&1&0&0&0&1&1&1&1&1&0&0&0&1&1&1&0&0&1&0&1\\
0&1&1&1&1&1&1&0&1&0&1&0&1&0&1&0&1&0&1&1&1&1&1\\

\end{array}\right)
}
&
{
K_{24}^{\ast} = \left(\begin{array}{cccccccccccccccccccccccc}
0&0&0&0&0&0&0&0&0&0&0&0&0&0&0&0&0&0&0&0&0&1&0&0\\
0&0&0&0&0&0&0&0&0&0&0&0&0&0&0&0&0&0&0&0&0&0&1&1\\
0&0&0&0&0&0&0&0&0&0&0&0&0&0&0&0&0&0&0&1&0&1&0&0\\
0&0&0&0&0&1&0&0&0&0&0&1&0&0&0&0&0&0&0&0&0&0&0&0\\
0&1&0&0&0&0&0&0&0&0&0&0&0&0&0&0&1&0&0&0&0&0&0&0\\
0&1&0&0&0&0&0&0&0&0&1&0&0&0&0&0&0&0&0&0&0&0&0&0\\
0&0&1&0&0&0&0&0&1&0&0&0&1&0&0&0&0&0&0&0&1&0&0&0\\
0&0&0&0&1&1&0&0&0&1&1&0&0&0&0&0&0&0&0&0&0&0&0&0\\
0&0&0&0&0&1&0&1&0&0&0&0&0&1&1&0&0&0&0&0&0&0&0&0\\
0&0&0&0&1&0&0&0&0&0&1&0&0&0&0&0&1&0&0&0&0&1&0&0\\
0&1&0&0&1&0&0&0&0&0&0&0&0&1&0&0&1&0&0&0&0&0&0&0\\
0&0&0&0&1&0&1&0&0&0&0&0&0&0&1&0&0&1&1&0&0&0&1&0\\
0&0&1&0&0&0&1&0&0&0&0&0&0&0&0&0&1&1&0&1&0&0&1&0\\
0&1&0&0&1&0&1&1&0&0&0&0&0&0&0&0&1&0&1&0&0&0&0&0\\
1&1&0&0&0&0&0&0&0&0&0&0&0&0&0&1&0&1&1&1&0&1&1&0\\
1&1&0&0&0&0&0&0&0&0&0&0&1&0&1&0&1&0&1&1&1&0&0&0\\
0&0&0&0&1&0&0&0&0&0&1&0&1&0&1&1&0&0&1&1&0&0&1&0\\
0&1&0&0&0&0&0&0&1&0&1&1&0&1&1&0&0&0&1&0&0&1&0&0\\
0&1&0&0&0&0&1&1&0&0&1&0&0&1&0&0&1&1&1&1&0&0&1&0\\
1&1&0&0&0&0&0&1&0&1&1&1&1&0&1&1&1&0&0&0&1&1&0&0\\
1&1&0&0&0&1&1&0&0&0&1&1&0&1&0&1&0&1&0&0&1&0&1&1\\
1&0&1&0&0&0&0&1&0&0&1&0&1&1&0&0&0&1&1&1&1&1&0&1\\
1&0&1&1&0&1&0&0&1&1&1&1&1&1&1&1&0&0&0&1&0&1&1&1\\
1&1&0&0&1&0&1&1&1&0&0&0&1&0&1&1&1&1&1&1&1&0&1&1\\
\end{array}\right)
} 
\\
{
K_{25}^{\ast} = \left(\begin{array}{ccccccccccccccccccccccccc}
0&0&0&0&0&0&0&0&0&1&0&0&0&0&0&0&0&0&0&0&0&0&0&0&0\\
0&0&0&0&0&0&0&0&1&0&0&1&0&0&0&0&0&0&0&0&0&0&0&0&0\\
0&0&0&0&0&1&0&0&0&0&0&0&0&0&0&0&0&0&0&0&0&0&1&0&0\\
0&0&0&0&0&0&0&1&0&0&0&0&0&0&0&0&0&0&0&0&0&0&1&0&0\\
0&0&0&0&0&0&0&1&0&0&0&1&0&0&0&0&0&0&0&0&0&0&0&0&0\\
0&0&0&0&0&0&0&0&0&1&0&0&0&0&0&0&0&0&0&0&0&0&0&1&0\\
0&0&1&0&0&0&0&0&0&0&0&0&0&1&0&1&1&0&0&0&0&0&0&0&0\\
0&0&0&1&0&0&1&0&0&1&0&0&0&0&0&0&0&0&0&0&0&0&1&0&0\\
0&0&0&0&0&0&0&1&0&0&0&0&1&0&0&1&0&0&0&0&0&0&1&0&0\\
0&0&0&0&0&0&0&0&0&0&0&0&0&0&0&0&0&0&1&1&0&1&0&1&0\\
0&0&0&0&0&1&0&1&0&1&1&0&0&0&0&0&0&0&0&0&0&0&0&0&0\\
1&0&0&0&0&0&0&1&0&1&0&0&0&0&0&0&0&1&0&0&0&1&0&0&1\\
0&0&0&0&0&0&0&0&1&1&1&1&0&0&0&0&0&0&0&0&0&0&1&1&0\\
0&0&0&0&0&0&0&1&0&0&0&1&0&0&1&0&0&0&1&0&0&1&0&0&1\\
0&0&1&0&0&0&1&0&0&0&0&0&1&1&0&0&0&1&1&1&0&1&0&0&0\\
0&0&1&1&0&0&0&1&0&0&1&0&0&0&0&0&0&1&0&0&1&1&1&0&0\\
0&0&0&0&0&1&0&1&0&0&0&1&1&1&0&0&0&1&0&1&0&0&0&1&0\\
1&0&0&0&0&0&0&0&0&0&0&0&1&1&0&0&0&0&1&0&1&1&1&0&1\\
0&0&1&0&0&0&0&1&1&0&0&0&0&1&0&0&0&0&1&1&0&0&1&0&1\\
0&0&1&0&0&1&1&1&0&0&1&0&0&1&1&0&0&0&0&1&0&1&1&0&0\\
0&0&1&0&1&0&0&1&1&1&0&0&1&0&1&0&1&0&0&1&0&1&1&0&1\\
1&0&0&0&0&1&1&1&1&0&0&0&1&1&0&1&0&0&1&1&0&1&0&0&1\\
1&0&0&1&0&1&0&1&0&1&0&0&0&1&0&0&1&0&0&0&1&1&1&1&1\\
1&1&1&0&1&1&0&0&0&1&0&0&1&0&0&1&1&1&1&1&1&1&0&1&1\\
1&0&1&1&0&0&1&1&1&0&1&1&0&1&1&0&1&1&1&1&1&1&1&0&1\\
\end{array}\right)
} 
&
{
%\hspace{32mm}
K_{26}^{\ast} = \left(
\begin{array}{cccccccccccccccccccccccccc}
0&0&0&0&0&0&0&0&0&0&0&0&0&0&0&0&1&0&0&0&0&0&0&0&0&0\\
0&1&0&0&0&0&0&0&0&0&0&0&0&0&0&0&1&0&0&0&0&0&0&0&0&0\\
0&1&0&0&0&0&0&0&0&0&0&0&0&0&0&0&0&0&0&0&0&0&0&0&1&0\\
0&0&0&0&0&0&0&0&0&0&0&0&0&0&0&0&0&0&0&0&1&1&0&0&0&0\\
0&0&0&0&0&0&0&0&0&0&0&0&0&0&0&0&0&0&0&0&1&0&0&0&1&0\\
0&0&0&0&0&0&0&0&0&0&0&1&0&0&0&0&0&0&0&1&0&0&0&0&0&0\\
0&0&0&0&0&0&0&0&0&0&0&0&1&0&0&1&0&1&0&0&0&0&0&0&0&1\\
0&0&0&0&1&0&0&0&0&0&1&0&0&0&0&0&0&0&0&0&1&0&0&0&1&0\\
0&0&1&0&0&0&0&0&0&1&0&1&0&0&1&0&0&0&0&0&0&0&0&0&0&0\\
0&0&0&0&0&0&0&0&1&0&0&0&0&0&1&0&1&1&0&0&0&0&0&0&0&0\\
0&1&0&0&0&0&0&1&0&0&0&0&0&0&0&0&0&1&0&0&0&0&0&0&1&0\\
1&0&0&1&0&0&1&0&0&1&0&0&0&0&0&0&0&1&0&0&0&0&0&0&0&1\\
0&0&0&1&1&0&0&1&0&0&0&0&0&1&0&0&0&0&0&0&1&0&0&0&0&1\\
0&0&1&1&0&0&1&0&0&1&0&0&0&0&0&0&0&0&0&1&0&0&0&0&1&0\\
1&1&0&1&1&0&0&0&0&0&0&0&0&0&0&0&0&0&0&1&0&0&0&0&1&0\\
0&0&0&0&0&0&1&0&1&0&0&1&0&0&0&0&0&1&1&0&1&0&0&1&0&1\\
0&0&0&1&1&1&0&0&0&0&1&0&1&1&0&0&0&0&1&0&0&0&0&0&1&0\\
1&0&0&1&1&0&0&0&0&0&0&0&0&1&0&0&0&1&0&1&0&1&0&1&0&0\\
1&1&1&1&0&1&0&0&0&0&0&1&1&0&0&0&0&1&0&1&0&0&0&0&0&1\\
0&1&0&0&1&1&1&0&0&0&1&1&1&0&0&0&0&0&1&1&0&0&0&1&0&0\\
0&1&1&1&0&0&0&0&0&0&1&1&0&1&0&1&1&1&1&1&0&0&0&1&0&0\\
0&0&1&0&0&1&0&0&1&0&0&1&0&0&1&0&1&1&0&1&0&1&0&1&1&1\\
1&0&1&1&1&0&0&0&0&1&0&0&1&1&1&0&1&0&0&1&0&0&0&1&0&1\\
0&0&1&0&1&1&0&1&0&0&1&1&1&1&1&0&1&0&1&0&0&1&0&0&0&0\\
1&0&1&0&0&1&1&1&1&1&1&1&1&0&1&1&1&0&0&0&0&0&1&1&0&1\\
1&1&1&1&1&1&1&0&0&0&1&1&1&1&1&0&0&1&1&1&1&1&0&1&1&1\\
\end{array}\right)
}

\end{array}
$
}}
\caption[]{Polarization kernels with the best error exponent of sizes 23-26}
\label{fBestKernels21}
\end{figure*}

\begin{figure*}[ht]
\centering
\scalebox{0.8}{
%\parbox{1.8\textwidth}
{ 
$
\arraycolsep=1.2pt\def\arraystretch{0.6}
\begin{array}{cc}
K_{27}^{\ast} = \left(\begin{array}{ccccccccccccccccccccccccccc}
0&0&0&0&0&0&0&0&0&0&0&0&0&0&0&0&0&0&0&0&0&0&0&1&0&0&0\\
0&1&0&0&0&1&0&0&0&0&0&0&0&0&0&0&0&0&0&0&0&0&0&0&0&0&0\\
0&0&0&0&0&0&0&0&0&0&1&0&0&0&0&0&0&0&0&0&0&0&0&0&1&0&0\\
0&0&0&0&0&0&0&0&0&0&0&0&0&0&1&0&0&0&0&0&0&0&0&1&0&0&0\\
0&0&0&0&0&0&0&0&0&0&0&0&1&0&0&0&0&0&0&0&0&0&0&1&0&0&0\\
0&0&0&0&0&0&0&0&0&0&0&0&0&0&0&0&0&0&0&0&0&0&1&1&0&0&0\\
0&0&0&0&0&0&0&0&0&1&0&0&0&0&1&1&0&0&0&0&0&0&1&0&0&0&0\\
0&1&0&0&0&0&0&0&0&0&0&0&0&0&0&0&0&1&0&1&1&0&0&0&0&0&0\\
0&0&1&0&1&0&0&0&0&0&0&0&1&0&0&0&1&0&0&0&0&0&0&0&0&0&0\\
0&0&0&0&1&0&1&0&0&0&0&1&0&0&0&0&1&0&0&0&0&0&0&0&0&0&0\\
0&0&0&0&1&1&0&0&0&0&0&0&0&0&1&0&0&0&0&0&0&0&1&0&0&0&0\\
1&1&0&0&0&0&1&0&0&1&1&0&0&1&0&0&0&0&0&0&0&0&0&0&0&0&0\\
0&0&0&0&0&0&0&0&0&1&0&0&0&1&1&0&0&0&0&0&0&0&1&1&0&1&0\\
0&1&0&0&1&1&0&0&0&0&0&0&0&0&0&0&0&0&0&0&0&1&0&0&1&1&0\\
0&1&0&1&0&0&1&0&0&1&0&0&1&0&0&0&0&0&0&0&0&0&0&1&0&0&0\\
0&0&0&0&0&1&0&0&0&0&0&0&1&1&1&1&0&0&0&0&0&1&0&1&0&1&0\\
0&0&0&0&0&0&0&1&0&1&0&0&1&0&0&0&0&1&0&0&0&0&1&0&1&1&1\\
0&0&0&0&0&0&1&0&1&0&1&0&1&1&1&0&0&0&0&0&0&1&1&0&0&0&0\\
0&0&0&0&1&0&0&0&0&1&1&0&0&0&0&0&0&0&1&1&0&1&1&0&1&1&1\\
0&1&1&0&0&0&0&1&0&1&0&1&0&0&0&1&1&0&0&0&0&1&0&0&1&1&0\\
1&1&1&1&0&0&1&1&0&0&0&0&0&0&0&1&0&1&0&0&0&0&0&0&1&0&1\\
0&0&0&1&0&0&0&1&0&1&1&0&0&1&1&1&0&0&0&0&1&1&1&0&1&1&0\\
0&1&0&1&0&0&0&0&1&0&1&1&1&0&0&0&1&0&0&0&0&1&1&0&1&1&1\\
1&0&1&1&1&0&1&0&1&1&1&1&0&1&0&0&0&0&0&0&0&1&1&0&0&0&0\\
0&0&1&0&0&0&1&1&1&0&0&0&0&1&0&1&1&0&0&0&0&1&1&1&1&0&1\\
1&0&0&0&1&0&1&0&1&0&1&1&1&0&0&1&0&1&0&1&1&1&1&1&1&0&1\\
1&1&1&1&0&1&1&1&1&1&1&1&0&1&1&1&1&1&0&0&0&1&0&0&1&1&1\\
\end{array}\right)
&
{
K_{28}^{\ast} = \left(\begin{array}{cccccccccccccccccccccccccccc}
0&0&0&1&0&0&0&0&0&0&0&0&0&0&0&0&0&0&0&0&0&0&0&0&0&0&0&0\\
0&0&0&1&0&0&0&0&0&0&0&0&0&0&0&1&0&0&0&0&0&0&0&0&0&0&0&0\\
0&0&0&0&0&0&0&0&0&0&0&0&0&0&0&0&0&0&0&1&0&0&1&0&0&0&0&0\\
0&0&0&0&0&0&0&0&0&0&0&0&0&0&0&1&0&0&0&0&0&0&1&0&0&0&0&0\\
0&0&0&0&0&0&0&0&0&0&0&0&0&0&0&0&0&0&0&0&0&1&1&0&0&0&0&0\\
0&0&1&1&0&0&0&0&0&0&0&0&0&0&0&0&0&0&0&0&0&0&0&0&0&0&0&0\\
0&0&1&0&0&0&1&0&0&0&0&1&0&0&0&0&0&0&0&0&1&0&0&0&0&0&0&0\\
0&0&0&1&1&0&0&0&0&0&0&0&0&0&0&0&0&0&0&0&1&0&1&0&0&0&0&0\\
0&0&0&1&0&0&0&0&0&0&0&1&0&0&0&1&0&0&0&0&0&1&0&0&0&0&0&0\\
0&0&0&0&0&1&0&0&0&0&0&1&0&0&0&1&0&0&0&0&0&0&1&0&0&0&0&0\\
0&0&0&0&0&0&1&0&1&1&0&0&0&0&0&0&0&0&0&1&0&0&0&0&0&0&0&0\\
0&0&0&0&0&1&0&0&0&0&0&1&0&0&1&0&0&0&0&1&1&0&0&0&0&0&1&0\\
1&0&1&1&0&0&0&0&1&1&0&0&0&0&0&0&0&0&0&0&1&0&0&0&0&0&0&0\\
0&0&0&0&0&1&0&0&1&1&0&1&0&0&1&1&0&0&0&0&0&0&0&0&0&0&0&0\\
0&1&0&0&0&0&0&1&0&1&0&1&0&0&0&1&0&0&0&0&0&0&0&0&0&0&1&0\\
0&1&0&0&1&1&0&0&0&0&0&1&0&0&1&0&0&0&0&0&0&1&0&0&0&0&0&0\\
0&0&0&1&1&1&0&0&0&1&0&0&1&0&0&0&0&0&1&0&1&0&1&0&0&0&0&0\\
0&1&1&0&0&0&0&0&0&0&0&1&0&0&1&1&1&0&0&0&1&0&1&0&0&0&0&0\\
1&1&0&0&1&1&1&1&1&1&0&0&0&0&0&0&0&0&0&0&0&0&0&0&0&0&0&0\\
0&0&1&0&1&0&0&0&1&0&0&0&1&0&0&1&1&0&0&0&1&1&0&0&1&0&1&0\\
0&0&1&1&0&0&1&1&1&0&0&1&0&0&0&0&0&0&0&1&1&0&1&0&0&0&1&0\\
1&0&0&0&1&1&0&1&1&0&1&1&1&0&1&0&0&0&0&0&0&1&0&0&0&0&0&0\\
1&1&0&0&0&0&1&1&0&0&1&1&1&0&0&1&0&0&0&0&0&1&0&0&1&1&1&0\\
0&1&0&0&1&1&1&0&1&1&1&1&0&0&0&1&1&0&0&0&0&1&1&0&0&0&0&0\\
0&0&1&0&0&1&0&1&0&1&1&0&1&0&1&1&1&0&1&0&0&1&1&1&1&0&0&0\\
1&1&1&0&1&0&0&1&0&1&1&1&0&1&0&1&0&0&1&0&1&0&1&0&0&0&1&0\\
0&0&1&0&1&1&1&0&1&0&0&1&0&1&1&1&0&0&1&1&1&1&1&1&0&0&0&1\\
1&1&1&1&1&1&1&1&1&1&1&1&1&0&1&1&1&1&0&1&1&1&1&0&1&1&1&0\\
\end{array}\right)
} 
\\
{
%\hspace{18mm}
K_{29}^{\ast} = \left(
\begin{array}{ccccccccccccccccccccccccccccc}
0&0&1&0&0&0&0&0&0&0&0&0&0&0&0&0&0&0&0&0&0&0&0&0&0&0&0&0&0\\
0&0&0&0&0&0&0&0&0&0&0&0&0&0&0&0&0&0&0&1&0&1&0&0&0&0&0&0&0\\
0&0&0&1&0&0&0&0&0&0&0&0&0&0&0&1&0&0&0&0&0&0&0&0&0&0&0&0&0\\
0&0&0&0&0&1&0&0&0&0&0&0&0&0&0&0&0&0&0&0&0&0&0&0&1&0&0&0&0\\
0&0&0&0&0&0&0&0&0&0&0&0&0&0&0&0&0&0&0&1&0&0&0&1&0&0&0&0&0\\
0&1&0&0&0&0&0&0&0&0&0&0&0&0&0&0&1&0&0&0&0&0&0&0&0&0&0&0&0\\
0&0&1&0&0&0&0&0&1&0&0&0&1&0&0&0&0&0&0&0&0&0&0&0&1&0&0&0&0\\
0&0&0&0&0&0&1&0&1&0&1&0&0&0&0&0&0&0&0&0&0&0&0&1&0&0&0&0&0\\
0&1&0&0&0&1&0&0&0&0&0&0&0&0&0&0&0&0&1&0&0&0&0&0&1&0&0&0&0\\
0&0&1&0&0&0&0&0&0&0&0&1&0&0&0&0&0&0&0&0&0&1&0&0&0&0&1&0&0\\
0&0&0&0&0&0&0&0&0&0&0&0&0&0&0&0&0&0&0&0&0&0&0&0&1&1&1&1&0\\
0&0&0&1&0&0&0&0&1&1&1&0&0&0&0&0&0&0&0&0&0&0&0&0&0&0&0&0&0\\
0&0&0&1&0&0&0&0&0&1&0&0&0&1&0&0&0&1&0&0&1&0&0&0&0&0&0&1&0\\
1&0&0&0&0&0&0&0&1&1&0&1&0&0&0&1&1&0&0&0&0&0&0&0&0&0&0&0&0\\
0&1&0&0&0&1&1&0&0&0&0&0&0&0&0&0&0&0&0&1&0&0&0&1&1&0&0&0&0\\
0&0&0&0&0&0&0&0&0&0&0&0&0&0&0&0&1&0&1&1&0&0&0&1&1&1&0&0&0\\
1&1&0&0&0&0&1&0&0&1&0&0&0&1&1&1&0&0&1&0&0&0&0&0&0&0&0&0&0\\
0&0&0&1&0&0&0&0&1&1&1&0&0&0&0&0&0&0&0&1&1&0&1&0&0&1&0&0&0\\
0&0&0&0&0&0&0&0&0&0&0&0&0&0&0&0&0&0&1&1&1&1&0&1&1&0&1&1&0\\
0&0&0&0&1&0&0&0&0&0&1&0&0&1&1&1&0&0&0&1&0&0&1&0&1&0&1&1&0\\
0&0&0&1&0&1&0&0&0&1&1&0&0&0&1&0&1&1&0&0&0&0&1&0&1&1&0&0&0\\
1&0&0&1&0&0&0&0&0&0&1&0&1&0&0&1&0&0&1&1&0&0&1&0&0&0&0&1&1\\
1&1&1&0&0&1&0&0&1&0&0&0&0&0&1&0&1&0&0&0&1&0&0&1&0&1&0&1&1\\
0&1&0&1&0&1&0&1&0&0&1&0&0&0&1&0&1&1&0&1&1&1&0&0&0&0&0&1&0\\
0&0&0&0&0&1&1&1&1&0&0&1&1&0&1&0&1&1&1&1&0&1&0&1&0&0&0&0&1\\
0&1&1&1&0&0&0&1&0&0&0&0&0&1&1&0&1&0&1&0&1&1&0&1&1&0&1&0&1\\
1&1&0&1&1&0&1&0&0&0&1&1&1&1&0&0&1&0&1&1&0&1&1&0&0&0&1&0&1\\
1&1&1&1&0&0&1&1&1&1&1&0&1&0&1&1&1&1&0&0&1&0&0&0&0&0&0&0&1\\
1&0&0&1&0&1&0&1&0&0&1&0&1&0&1&1&0&1&1&1&1&1&1&1&1&1&1&1&1\\
\end{array}\right)
}
\end{array}
$
}}
\caption[]{Polarization kernels with the best error exponent of sizes 27-29}
\label{fBestKernels27}
\end{figure*}

\begin{figure*}[ht]
\centering
\scalebox{0.8}{
%\parbox{1.2\textwidth}
{ 
$
\arraycolsep=1.2pt\def\arraystretch{0.6}
\begin{array}{cc}
{
K_{18} = \left(\begin{array}{cccccccccccccccccc}
1&0&0&0&0&0&0&0&0&0&0&0&0&0&0&0&0&0\\
1&0&0&0&0&0&0&0&1&0&0&0&0&0&0&0&0&0\\
1&0&1&0&0&0&0&0&0&0&0&0&0&0&0&0&0&0\\
1&0&1&0&0&0&0&0&1&0&0&0&1&0&0&0&0&0\\
0&0&0&0&0&0&0&0&1&0&1&0&0&0&0&0&0&0\\
0&0&1&1&0&0&0&0&0&0&0&0&0&0&0&0&0&0\\
1&1&0&0&0&0&0&0&0&0&0&0&0&0&0&0&0&0\\
1&0&1&1&0&0&1&0&0&0&0&0&0&0&0&0&0&0\\
1&1&0&1&0&1&0&0&0&0&0&0&0&0&0&0&0&0\\
0&0&1&0&0&1&1&0&1&0&1&0&0&0&0&0&1&0\\
1&1&1&0&1&0&0&0&0&0&0&0&0&0&0&0&0&0\\
1&0&0&1&1&0&1&0&1&1&0&0&0&0&0&0&0&0\\
0&1&1&0&1&0&1&0&1&1&0&0&0&0&0&0&1&1\\
1&0&1&0&1&0&1&0&1&0&1&0&1&0&1&0&0&0\\
1&1&0&0&1&1&0&0&1&1&0&0&1&1&0&0&0&0\\
1&1&1&1&0&0&0&0&1&1&1&1&0&0&0&0&0&0\\
1&1&1&1&1&1&1&1&0&0&0&0&0&0&0&0&0&0\\
1&1&1&1&1&1&1&1&1&1&1&1&1&1&1&1&0&0\\
\end{array}\right)
}
&
{
K_{20} = \left(\begin{array}{cccccccccccccccccccc}
1&0&0&0&0&0&0&0&0&0&0&0&0&0&0&0&0&0&0&0\\
1&0&0&0&0&0&0&0&0&0&0&0&1&0&0&0&0&0&0&0\\
1&0&0&0&0&0&0&0&1&0&0&0&0&0&0&0&0&0&0&0\\
1&0&0&0&0&0&0&0&1&0&0&0&1&0&0&0&1&0&0&0\\
1&0&0&0&1&0&0&0&0&0&0&0&0&0&0&0&0&0&0&0\\
1&0&0&0&1&0&0&0&1&1&0&0&0&0&0&0&0&0&0&0\\
1&0&1&0&0&0&0&0&0&0&0&0&0&0&0&0&0&0&0&0\\
1&1&0&0&0&0&0&0&0&0&0&0&0&0&0&0&0&0&0&0\\
0&1&1&0&1&0&1&0&0&0&0&0&0&0&0&0&0&0&0&0\\
1&1&0&0&1&1&0&0&0&0&0&0&0&0&0&0&0&0&0&0\\
1&1&0&0&0&1&1&0&1&0&1&0&0&0&0&0&0&0&0&0\\
1&0&1&0&0&1&1&0&1&1&0&0&0&0&0&0&1&0&1&0\\
0&1&1&0&1&1&0&0&1&0&1&0&0&0&0&0&1&1&0&0\\
1&0&1&0&1&0&1&0&1&0&1&0&1&0&1&0&0&0&0&0\\
1&1&1&1&0&0&0&0&0&0&0&0&0&0&0&0&0&0&0&0\\
0&0&1&1&1&1&0&0&1&1&0&0&1&1&0&0&0&0&0&0\\
1&1&0&0&1&1&0&0&1&1&0&0&1&1&0&0&1&1&1&1\\
1&1&1&1&0&0&0&0&1&1&1&1&0&0&0&0&0&0&0&0\\
1&1&1&1&1&1&1&1&0&0&0&0&0&0&0&0&0&0&0&0\\
1&1&1&1&1&1&1&1&1&1&1&1&1&1&1&1&0&0&0&0\\
\end{array}\right)
}
\\
{
K_{24}^{(r)} = \left(
\begin{array}{cccccccccccccccccccccccc}
1&0&0&0&0&0&0&0&0&0&0&0&0&0&0&0&0&0&0&0&0&0&0&0\\
1&0&0&0&0&0&0&0&0&0&0&0&0&0&0&0&1&0&0&0&0&0&0&0\\
1&0&0&0&1&0&0&0&0&0&0&0&0&0&0&0&0&0&0&0&0&0&0&0\\
1&0&0&0&0&0&0&0&1&0&0&0&0&0&0&0&0&0&0&0&0&0&0&0\\
1&0&1&0&0&0&0&0&0&0&0&0&0&0&0&0&0&0&0&0&0&0&0&0\\
1&1&0&0&0&0&0&0&0&0&0&0&0&0&0&0&0&0&0&0&0&0&0&0\\
1&0&0&0&1&0&0&0&1&0&0&0&1&0&0&0&0&0&0&0&0&0&0&0\\
1&0&1&0&0&0&0&0&1&0&1&0&0&0&0&0&0&0&0&0&0&0&0&0\\
1&1&0&0&0&0&0&0&1&1&0&0&0&0&0&0&0&0&0&0&0&0&0&0\\
1&0&1&0&1&0&1&0&0&0&0&0&0&0&0&0&0&0&0&0&0&0&0&0\\
1&1&0&0&1&1&0&0&0&0&0&0&0&0&0&0&0&0&0&0&0&0&0&0\\
1&1&1&1&0&0&0&0&0&0&0&0&0&0&0&0&0&0&0&0&0&0&0&0\\
0&0&0&0&0&1&1&0&1&1&1&0&1&0&0&0&1&0&0&0&1&0&0&0\\
0&1&1&1&0&0&1&0&1&0&0&0&1&0&0&0&1&0&1&0&0&0&0&0\\
1&0&1&0&1&0&1&0&1&1&0&0&0&0&0&0&1&1&0&0&0&0&0&0\\
0&1&1&0&1&0&1&0&0&1&1&0&1&0&1&0&0&0&0&0&0&0&0&0\\
0&0&1&1&1&1&0&0&1&1&0&0&1&1&0&0&0&0&0&0&0&0&0&0\\
1&1&0&0&1&1&0&0&1&1&1&1&0&0&0&0&0&0&0&0&0&0&0&0\\
1&1&1&1&1&1&1&1&0&0&0&0&0&0&0&0&0&0&0&0&0&0&0&0\\
1&0&1&0&1&0&1&0&1&0&1&0&1&0&1&0&1&0&1&0&1&0&1&0\\
1&1&0&0&1&1&0&0&1&1&0&0&1&1&0&0&1&1&0&0&1&1&0&0\\
1&1&1&1&0&0&0&0&1&1&1&1&0&0&0&0&1&1&1&1&0&0&0&0\\
1&1&1&1&1&1&1&1&0&0&0&0&0&0&0&0&1&1&1&1&1&1&1&1\\
1&1&1&1&1&1&1&1&1&1&1&1&1&1&1&1&0&0&0&0&0&0&0&0\\
\end{array}\right)
} 
&
{
K_{24} = \left(
\begin{array}{cccccccccccccccccccccccc}
1&0&0&0&0&0&0&0&0&0&0&0&0&0&0&0&0&0&0&0&0&0&0&0\\
1&0&0&0&0&0&0&0&0&0&0&0&0&0&0&0&1&0&0&0&0&0&0&0\\
1&0&0&0&0&0&0&0&1&0&0&0&0&0&0&0&0&0&0&0&0&0&0&0\\
1&0&0&0&0&0&0&0&1&0&0&0&0&0&0&0&1&0&0&0&1&0&0&0\\
1&0&1&0&0&0&0&0&0&0&0&0&0&0&0&0&0&0&0&0&0&0&0&0\\
1&0&1&0&0&0&0&0&1&0&0&0&1&0&0&0&0&0&0&0&0&0&0&0\\
1&0&0&0&1&0&0&0&0&0&0&0&0&0&0&0&0&0&0&0&0&0&0&0\\
1&0&0&0&1&0&0&0&1&0&1&0&0&0&0&0&0&0&0&0&0&0&0&0\\
1&0&1&0&0&0&0&0&1&0&1&0&0&0&0&0&1&0&1&0&0&0&0&0\\
1&1&0&0&0&0&0&0&0&0&0&0&0&0&0&0&0&0&0&0&0&0&0&0\\
1&1&0&0&0&0&0&0&1&1&0&0&0&0&0&0&0&0&0&0&0&0&0&0\\
1&0&1&0&1&0&1&0&0&0&0&0&0&0&0&0&0&0&0&0&0&0&0&0\\
1&0&1&0&1&0&1&0&1&1&0&0&0&0&0&0&1&1&0&0&0&0&0&0\\
0&1&1&0&1&0&1&0&0&1&1&0&1&0&1&0&0&0&0&0&0&0&0&0\\
1&0&1&0&1&0&1&0&1&0&1&0&1&0&1&0&1&0&1&0&1&0&1&0\\
1&1&0&0&1&1&0&0&0&0&0&0&0&0&0&0&0&0&0&0&0&0&0&0\\
1&1&1&1&0&0&0&0&0&0&0&0&0&0&0&0&0&0&0&0&0&0&0&0\\
0&0&1&1&1&1&0&0&1&1&0&0&1&1&0&0&0&0&0&0&0&0&0&0\\
1&1&0&0&1&1&0&0&1&1&1&1&0&0&0&0&0&0&0&0&0&0&0&0\\
1&1&0&0&1&1&0&0&1&1&0&0&1&1&0&0&1&1&0&0&1&1&0&0\\
1&1&1&1&0&0&0&0&1&1&1&1&0&0&0&0&1&1&1&1&0&0&0&0\\
1&1&1&1&1&1&1&1&0&0&0&0&0&0&0&0&0&0&0&0&0&0&0&0\\
1&1&1&1&1&1&1&1&0&0&0&0&0&0&0&0&1&1&1&1&1&1&1&1\\
1&1&1&1&1&1&1&1&1&1&1&1&1&1&1&1&0&0&0&0&0&0&0&0\\
\end{array}\right)
}
\\
{
K_{27} = \left(
\begin{array}{ccccccccccccccccccccccccccc}
1&0&0&0&0&0&0&0&0&0&0&0&0&0&0&0&0&0&0&0&0&0&0&0&0&0&0\\
1&0&0&0&0&0&0&0&0&0&0&0&1&0&0&0&0&0&0&0&0&0&0&0&0&0&0\\
1&0&0&0&0&0&0&0&1&0&0&0&0&0&0&0&0&0&0&0&0&0&0&0&0&0&0\\
1&0&0&0&0&0&0&0&1&0&0&0&1&0&0&0&0&0&0&0&0&0&0&0&1&0&0\\
1&0&0&0&1&0&0&0&0&0&0&0&0&0&0&0&0&0&0&0&0&0&0&0&0&0&0\\
1&0&1&0&0&0&0&0&0&0&0&0&0&0&0&0&0&0&0&0&0&0&0&0&0&0&0\\
1&0&1&0&0&0&0&0&0&0&0&0&1&0&0&0&1&0&0&0&0&0&0&0&0&0&0\\
1&0&0&0&1&0&0&0&1&0&1&0&0&0&0&0&0&0&0&0&0&0&0&0&0&0&0\\
1&0&0&0&1&0&0&0&1&0&0&0&1&0&0&0&1&0&0&0&1&0&0&0&0&0&0\\
1&1&0&0&0&0&0&0&0&0&0&0&0&0&0&0&0&0&0&0&0&0&0&0&0&0&0\\
1&1&0&0&0&0&0&0&1&1&0&0&0&0&0&0&0&0&0&0&0&0&0&0&0&0&0\\
1&0&1&0&1&0&1&0&0&0&0&0&0&0&0&0&0&0&0&0&0&0&0&0&0&0&0\\
0&0&0&0&0&1&1&0&1&0&1&0&0&0&0&0&1&0&1&0&0&0&0&0&0&0&0\\
1&0&1&0&0&0&0&0&1&0&1&0&1&0&1&0&0&0&0&0&0&0&0&0&0&0&0\\
1&0&1&0&0&0&0&0&1&0&1&0&0&0&0&0&1&0&1&0&0&0&0&0&1&0&1\\
1&0&1&0&1&0&1&0&1&1&0&0&0&0&0&0&1&1&0&0&0&0&0&0&0&0&0\\
0&1&1&0&1&0&1&0&0&1&1&0&1&0&1&0&0&0&0&0&0&0&0&0&1&1&0\\
1&0&1&0&1&0&1&0&1&0&1&0&1&0&1&0&1&0&1&0&1&0&1&0&0&0&0\\
1&1&0&0&1&1&0&0&0&0&0&0&0&0&0&0&0&0&0&0&0&0&0&0&0&0&0\\
1&1&1&1&0&0&0&0&0&0&0&0&0&0&0&0&0&0&0&0&0&0&0&0&0&0&0\\
0&0&1&1&1&1&0&0&1&1&0&0&1&1&0&0&0&0&0&0&0&0&0&0&0&0&0\\
1&1&0&0&1&1&0&0&1&1&1&1&0&0&0&0&0&0&0&0&0&0&0&0&0&0&0\\
1&1&0&0&1&1&0&0&1&1&0&0&1&1&0&0&1&1&0&0&1&1&0&0&0&0&0\\
1&1&1&1&0&0&0&0&1&1&1&1&0&0&0&0&1&1&1&1&0&0&0&0&0&0&0\\
1&1&1&1&1&1&1&1&0&0&0&0&0&0&0&0&0&0&0&0&0&0&0&0&0&0&0\\
1&1&1&1&1&1&1&1&0&0&0&0&0&0&0&0&1&1&1&1&1&1&1&1&0&0&0\\
1&1&1&1&1&1&1&1&1&1&1&1&1&1&1&1&0&0&0&0&0&0&0&0&0&0&0\\
\end{array}\right)
}
&
{
K_{32}^{(r)} = \left(\begin{array}{cccccccccccccccccccccccccccccccc}
1&0&0&0&0&0&0&0&0&0&0&0&0&0&0&0&0&0&0&0&0&0&0&0&0&0&0&0&0&0&0&0\\
1&0&0&0&0&0&0&0&0&0&0&0&0&0&0&0&1&0&0&0&0&0&0&0&0&0&0&0&0&0&0&0\\
1&0&0&0&0&0&0&0&1&0&0&0&0&0&0&0&0&0&0&0&0&0&0&0&0&0&0&0&0&0&0&0\\
1&0&0&0&0&0&0&0&1&0&0&0&0&0&0&0&1&0&0&0&0&0&0&0&1&0&0&0&0&0&0&0\\
1&0&0&0&1&0&0&0&0&0&0&0&0&0&0&0&0&0&0&0&0&0&0&0&0&0&0&0&0&0&0&0\\
1&0&0&0&1&0&0&0&0&0&0&0&0&0&0&0&1&0&0&0&1&0&0&0&0&0&0&0&0&0&0&0\\
1&0&1&0&0&0&0&0&0&0&0&0&0&0&0&0&0&0&0&0&0&0&0&0&0&0&0&0&0&0&0&0\\
1&0&1&0&0&0&0&0&0&0&0&0&0&0&0&0&1&0&1&0&0&0&0&0&0&0&0&0&0&0&0&0\\
0&0&1&0&1&0&0&0&1&0&0&0&1&0&0&0&1&0&1&0&0&0&0&0&0&0&0&0&0&0&0&0\\
1&0&0&0&1&0&0&0&1&0&0&0&1&0&0&0&1&0&0&0&1&0&0&0&1&0&0&0&1&0&0&0\\
1&1&0&0&0&0&0&0&0&0&0&0&0&0&0&0&0&0&0&0&0&0&0&0&0&0&0&0&0&0&0&0\\
1&1&0&0&0&0&0&0&0&0&0&0&0&0&0&0&1&1&0&0&0&0&0&0&0&0&0&0&0&0&0&0\\
0&1&1&0&0&0&0&0&1&0&1&0&0&0&0&0&1&1&0&0&0&0&0&0&0&0&0&0&0&0&0&0\\
1&0&1&0&0&0&0&0&1&0&1&0&0&0&0&0&1&0&1&0&0&0&0&0&1&0&1&0&0&0&0&0\\
1&0&1&0&1&0&1&0&0&0&0&0&0&0&0&0&0&0&0&0&0&0&0&0&0&0&0&0&0&0&0&0\\
0&1&1&0&1&0&1&0&1&1&0&0&0&0&0&0&0&0&0&0&0&0&0&0&0&0&0&0&0&0&0&0\\
1&1&0&0&0&0&0&0&1&1&0&0&0&0&0&0&1&1&0&0&0&0&0&0&1&1&0&0&0&0&0&0\\
0&1&1&0&1&0&1&0&1&1&0&0&0&0&0&0&0&1&1&0&1&0&1&0&1&1&0&0&0&0&0&0\\
1&1&0&0&1&1&0&0&0&0&0&0&0&0&0&0&0&0&0&0&0&0&0&0&0&0&0&0&0&0&0&0\\
1&1&0&0&1&1&0&0&0&0&0&0&0&0&0&0&1&1&0&0&1&1&0&0&0&0&0&0&0&0&0&0\\
0&1&1&0&0&1&1&0&1&0&1&0&1&0&1&0&1&1&0&0&1&1&0&0&0&0&0&0&0&0&0&0\\
1&0&1&0&1&0&1&0&1&0&1&0&1&0&1&0&1&0&1&0&1&0&1&0&1&0&1&0&1&0&1&0\\
1&1&1&1&0&0&0&0&0&0&0&0&0&0&0&0&0&0&0&0&0&0&0&0&0&0&0&0&0&0&0&0\\
1&1&1&1&0&0&0&0&0&0&0&0&0&0&0&0&1&1&1&1&0&0&0&0&0&0&0&0&0&0&0&0\\
0&0&1&1&1&1&0&0&1&1&0&0&1&1&0&0&1&1&1&1&0&0&0&0&0&0&0&0&0&0&0&0\\
1&1&0&0&1&1&0&0&1&1&0&0&1&1&0&0&1&1&0&0&1&1&0&0&1&1&0&0&1&1&0&0\\
1&1&1&1&0&0&0&0&1&1&1&1&0&0&0&0&0&0&0&0&0&0&0&0&0&0&0&0&0&0&0&0\\
1&1&1&1&0&0&0&0&1&1&1&1&0&0&0&0&1&1&1&1&0&0&0&0&1&1&1&1&0&0&0&0\\
1&1&1&1&1&1&1&1&0&0&0&0&0&0&0&0&0&0&0&0&0&0&0&0&0&0&0&0&0&0&0&0\\
1&1&1&1&1&1&1&1&0&0&0&0&0&0&0&0&1&1&1&1&1&1&1&1&0&0&0&0&0&0&0&0\\
1&1&1&1&1&1&1&1&1&1&1&1&1&1&1&1&0&0&0&0&0&0&0&0&0&0&0&0&0&0&0&0\\
1&1&1&1&1&1&1&1&1&1&1&1&1&1&1&1&1&1&1&1&1&1&1&1&1&1&1&1&1&1&1&1\\

\end{array}\right)
}
\end{array}
$
}}
\caption[]{Polarization kernels which admit low complexity processing}
\label{fSimpleKernels}
\end{figure*}

%\bibliographystyle{ieeetran}
%\bibliography{coding,comm,math,misc,trifonov,miloslavskaya,trofimiuk}

% Generated by IEEEtran.bst, version: 1.14 (2015/08/26)

\begin{IEEEbiography}[{\includegraphics[angle=270,width=\textwidth]{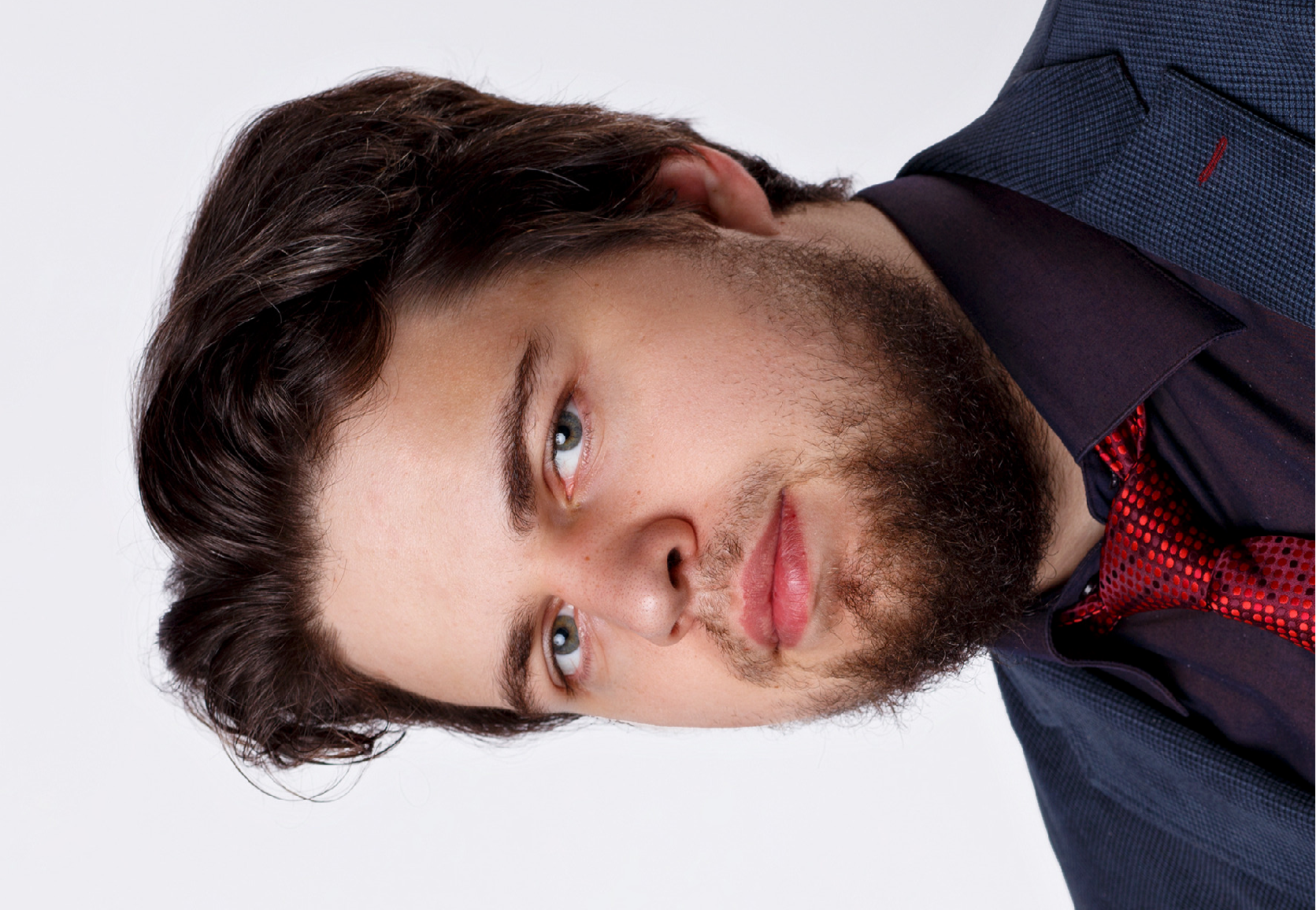}}]
{Grigorii Trofimiuk} was born in Boksitogorsk, Russia, in 1994. He received the B.Sc. and M.Sc. degrees in computer science from St.Petersburg Polytechnic University, in 2016 and 2018, respectively, and Ph.D. (Candidate of Science) degree in computer science from ITMO University in 2022. He is currently a Post-Doctoral Research Associate in telecommunications with the Laboratory of Information and Coding Theory, ITMO University, St.Petersburg, Russia. His research interest includes coding theory and its applications in telecommunications.
\end{IEEEbiography}

\end{document}